\documentclass[11pt]{article}%
\usepackage{amssymb}
\usepackage{amsfonts}
\usepackage{amsmath}
\usepackage{amssymb}
\usepackage{url}
\usepackage{hyperref}
\usepackage{graphicx}
\usepackage{bm}
\usepackage{color}
\usepackage{subfigure}%
\providecommand{\U}[1]{\protect\rule{.1in}{.1in}}
\providecommand{\U}[1]{\protect\rule{.1in}{.1in}}
\newtheorem{theo}{Theorem}[section]
\newtheorem{prop}[theo]{Proposition}
\newtheorem{lem}[theo]{Lemma}

\newtheorem{exam}[theo]{Example}

\numberwithin{equation}{section}

\begin{document}

\title{Forward Smoothing using Sequential Monte Carlo }
\author{Pierre Del Moral\thanks{Centre INRIA Bordeaux et Sud-Ouest \& Institut de
Math\'{e}matiques de Bordeaux , Universit\'{e} de Bordeaux I, 351 cours de la
Lib\'{e}ration 33405 Talence cedex, France (Pierre.Del-Moral@inria.fr)} ,
Arnaud Doucet\thanks{The Institute of Statistical Mathematics, 4-6-7
Minami-Azabu, Minato-ku, Tokyo 106-8569, Japan and Department of Statistics \&
Department of Computer Science, University of British Columbia, 333-6356
Agricultural Road, Vancouver, BC, V6T 1Z2, Canada (Arnaud@stat.ubc.ca)} ,
Sumeetpal S. Singh\thanks{Department of Engineering, University of Cambridge,
Trumpington Street, CB2 1PZ, United Kingdom (sss40@cam.ac.uk)}}
\date{Technical Report CUED/F-INFENG/TR 638, Cambridge University.\\
First version September 27th 2009, Second version November 24th 2009, Third
version December 12th 2010.}
\maketitle

\begin{abstract}
Sequential Monte Carlo (SMC) methods are a widely used set of computational
tools for inference in non-linear non-Gaussian state-space models. We propose
a new SMC algorithm to compute the expectation of additive functionals
recursively. Essentially, it is an online or \textquotedblleft
forward-only\textquotedblright\ implementation of a forward filtering backward
smoothing SMC algorithm proposed in \cite{Douc00}. Compared to the standard
path space SMC\ estimator whose asymptotic variance increases quadratically
with time even under favourable mixing assumptions, the asymptotic variance of
the proposed SMC\ estimator only increases linearly with time. This forward
smoothing procedure allows us to implement on-line maximum likelihood
parameter estimation algorithms which do not suffer from the particle path
degeneracy problem.

\emph{Some key words}: Expectation-Maximization, Forward Filtering Backward
Smoothing, Recursive Maximum Likelihood, Sequential Monte Carlo, Smoothing,
State-Space Models.

\end{abstract}

\section{Introduction\label{sec:intro}}

\subsection{State-space models and inference aims}

State-space models (SSM) are a very popular class of non-linear and
non-Gaussian time series models in statistics, econometrics and information
engineering; see for example \cite{cappe2005}, \cite{Douc01},
\cite{DurbinKoopman2001}. An SSM is comprised of a pair of discrete-time
stochastic processes, $\left\{  X_{n}\right\}  _{n\mathbb{\geq}0}$ and
$\left\{  Y_{n}\right\}  _{n\geq0}$, where the former is an $\mathcal{X}%
$-valued unobserved process and the latter is a $\mathcal{Y}$-valued process
which is observed. The hidden process $\left\{  X_{n}\right\}  _{n\mathbb{\geq
}0}$ is a\ Markov process with initial density $\mu_{\theta}\left(  x\right)
$ and Markov transition density $f_{\theta}\left(  \left.  x^{\prime
}\right\vert x\right)  $, i.e.%
\begin{equation}
X_{0}\sim\mu_{\theta}\left(  \cdot\right)  \text{ and }\left.  X_{n}%
\right\vert \left(  X_{n-1}=x_{n-1}\right)  \sim f_{\theta}\left(  \left.
\cdot\right\vert x_{n-1}\right)  ,\qquad n\geq1. \label{eq:evol}%
\end{equation}
It is assumed that the observations $\left\{  Y_{n}\right\}  _{n\geq0}%
$\ conditioned upon $\left\{  X_{n}\right\}  _{n\mathbb{\geq}0}$ are
statistically independent and have marginal density $g_{\theta}\left(  \left.
y\right\vert x\right)  $, i.e.%
\begin{equation}
\left.  Y_{n}\right\vert \left(  \left\{  X_{k}\right\}  _{k\geq0}=\left\{
x_{k}\right\}  _{k\geq0}\right)  \sim g_{\theta}\left(  \left.  \cdot
\right\vert x_{n}\right)  . \label{eq:obs}%
\end{equation}
We also assume that $\mu_{\theta}\left(  x\right)  $, $f_{\theta}\left(
\left.  x\right\vert x^{\prime}\right)  $ and $g_{\theta}\left(  \left.
y\right\vert x\right)  $ are densities with respect to (w.r.t.) suitable
dominating measures denoted generically as $dx$ and $dy$. For example, if
$\mathcal{X}\subseteq\mathbb{R}^{p}$ and $\mathcal{Y}\subseteq\mathbb{R}^{q}$
then the dominating measures could be the Lebesgue measures. The variable
$\theta$\ in the densities of these random variables are the particular
parameters of the model. The set of possible values for $\theta$ is denoted
$\Theta$. The model (\ref{eq:evol})-(\ref{eq:obs}) is also referred to as a
Hidden Markov Model (HMM) in the literature.

For any sequence $\left\{  z_{n}\right\}  _{n\in\mathbb{Z}}$ and integers
$j\geq i$, let $z_{i:j}$ denote the set $\left\{  z_{i},z_{i+1},...,z_{j}%
\right\}  $. (When $j<i$\ this is to be understood as the empty set.)
Equations (\ref{eq:evol}) and (\ref{eq:obs}) define the joint density of
$\left(  X_{0:n},Y_{0:n}\right)  $,%
\begin{equation}
p_{\theta}\left(  x_{0:n},y_{0:n}\right)  =\mu_{\theta}\left(  x_{0}\right)
\prod\limits_{k=1}^{n}f_{\theta}\left(  \left.  x_{k}\right\vert
x_{k-1}\right)  \prod\limits_{k=0}^{n}g_{\theta}\left(  \left.  y_{k}%
\right\vert x_{k}\right)  \label{eq:jointdensity}%
\end{equation}
which yields the marginal likelihood,
\begin{equation}
p_{\theta}\left(  y_{0:n}\right)  =\int p_{\theta}\left(  x_{0:n}%
,y_{0:n}\right)  dx_{0:n}. \label{eq:decompositionjointdistribution}%
\end{equation}
Let $s_{k}:\mathcal{X}\times\mathcal{X}\rightarrow\mathbb{R}$, $k\in
\mathbb{N}$, be a sequence of functions and $S_{n}:\mathcal{X}^{n}%
\rightarrow\mathbb{R}$, $n\in\mathbb{N}$, be the corresponding sequence of
additive functionals constructed from $s_{k}$ as
follows\footnote{Incorporating dependency of $s_{k}$\ on $y_{k}$, i.e. $S_{n}$
is the sums of term of the form $s_{k}\left(  x_{k-1},x_{k},y_{k}\right)  $,
is merely a matter of redefining $s_{k}$ in the computations to follow.}%
\begin{equation}
S_{n}\left(  x_{0:n}\right)  =\sum_{k=1}^{n}s_{k}\left(  x_{k-1},x_{k}\right)
. \label{eq:additiveFunctional}%
\end{equation}
\ There are many instances where it is necessary to be able to compute the
following expectations recursively in time,%
\begin{equation}
\mathcal{S}_{n}^{\theta}=\mathbb{E}_{\theta}\left[  \left.  S_{n}\left(
X_{0:n}\right)  \right\vert y_{0:n}\right]  .
\label{eq:representationadditivefunctional}%
\end{equation}
The conditioning implies the expectation should be computed w.r.t. the density
of $X_{0:n}$ given $Y_{0:n}=y_{0:n}$, i.e. $p_{\theta}\left(  \left.
x_{0:n}\right\vert y_{0:n}\right)  \propto p_{\theta}\left(  x_{0:n}%
,y_{0:n}\right)  $ and for this reason $\mathcal{S}_{n}^{\theta}$ is referred
to as a \emph{smoothed additive functional}.

As the first example of the need to perform such computations, consider the
problem of computing the score vector, $\nabla\log p_{\theta}\left(
y_{0:n}\right)  $. The score is a vector in $\mathbb{R}^{d}$ and its
$i^{\text{th}}$ component is%
\begin{equation}
\left[  \nabla\log p_{\theta}\left(  y_{0:n}\right)  \right]  _{i}%
=\frac{\partial\log\text{ }p_{\theta}\left(  y_{0:n}\right)  }{\partial
\theta^{i}}. \label{eq:score}%
\end{equation}
Using Fisher's identity, the problem of computing the score becomes an
instance of (\ref{eq:representationadditivefunctional}), i.e.%
\begin{align}
\nabla\log p_{\theta}\left(  y_{0:n}\right)   &  =\sum\limits_{k=1}%
^{n}\mathbb{E}_{\theta}\left[  \left.  \nabla\log f_{\theta}\left(  \left.
X_{k}\right\vert X_{k-1}\right)  \right\vert y_{0:n}\right]  +\sum
\limits_{k=0}^{n}\mathbb{E}_{\theta}\left[  \left.  \nabla\log g_{\theta
}\left(  \left.  y_{k}\right\vert X_{k}\right)  \right\vert y_{0:n}\right]
\nonumber\\
&  +\mathbb{E}_{\theta}\left[  \left.  \nabla\log\mu_{\theta}\left(
X_{0}\right)  \right\vert y_{0:n}\right]  . \label{eq:scoreAdditiveFunctional}%
\end{align}
An alternative representation of the score as a smoothed additive functional
based on infinitesimal perturbation analysis is given in \cite{coquelin2008}.
The score has applications to Maximum Likelihood (ML)\ parameter estimation
\cite{legland1997}, \cite{poyadjis2009}.

The second example is ML\ parameter estimation using the
Expectation-Maximization (EM) algorithm. Let $y_{0:n}$ be a batch of data and
the aim is to maximise $p_{\theta}\left(  y_{0:n}\right)  $ w.r.t. $\theta$.
Given a current estimate $\theta^{\prime}$, a new estimate $\theta
^{\prime\prime}$ is obtained by maximizing the function\
\begin{align*}
Q\left(  \theta^{\prime},\theta\right)   &  =\sum\limits_{k=1}^{n}%
\mathbb{E}_{\theta^{\prime}}\left[  \left.  \log f_{\theta}\left(  \left.
X_{k}\right\vert X_{k-1}\right)  \right\vert y_{0:n}\right]  +\sum
\limits_{k=0}^{n}\mathbb{E}_{\theta^{\prime}}\left[  \left.  \log g_{\theta
}\left(  \left.  y_{k}\right\vert X_{k}\right)  \right\vert y_{0:n}\right] \\
&  +\mathbb{E}_{\theta^{\prime}}\left[  \left.  \log\mu_{\theta}\left(
X_{0}\right)  \right\vert y_{0:n}\right]
\end{align*}
w.r.t. $\theta$ and setting $\theta^{\prime\prime}$\ to the maximising
argument. A fundamental property of the EM algorithm\ is $p_{\theta
^{\prime\prime}}\left(  y_{0:n}\right)  \geq p_{\theta^{\prime}}\left(
y_{0:n}\right)  $. For linear Gaussian models and finite state-space HMM, it
is possible to perform the computations in the definition of $Q\left(
\theta^{\prime},\theta\right)  $. For general non-linear non-Gaussian
state-space models of the form (\ref{eq:evol})-(\ref{eq:obs}), we need to rely
on numerical approximation schemes.

\subsection{Current approaches to smoothing with SMC\label{sec:litReview}}

SMC methods are a class of algorithms that sequentially\ approximate the
sequence of posterior distributions $\left\{  p_{\theta}\left(  dx_{0:n}%
|y_{0:n}\right)  \right\}  _{n\geq0}$ using a set of $N$ weighted random
samples called particles. Specifically, the SMC\ approximation of $p_{\theta
}\left(  dx_{0:n}|y_{0:n}\right)  $, for $n\geq0$, is
\begin{equation}
\widehat{p}_{\theta}\left(  dx_{0:n}|y_{0:n}\right)  :=\sum_{i=1}^{N}%
W_{n}^{\left(  i\right)  }\delta_{X_{0:n}^{\left(  i\right)  }}\left(
dx_{0:n}\right)  ,\text{ }W_{n}^{\left(  i\right)  }\geq0,\text{\ }\sum
_{i=1}^{N}W_{n}^{\left(  i\right)  }=1, \label{eq:filteringdistribution}%
\end{equation}
where $W_{n}^{\left(  i\right)  }$ is the importance weight associated to
particle $X_{0:n}^{\left(  i\right)  }$ and $\delta_{X_{0:n}^{\left(
i\right)  }}$\ is the Dirac measure with an atom at $X_{0:n}^{\left(
i\right)  }$. The particles are propagated forward in time using a combination
of importance sampling and resampling steps and there are several variants of
both these steps; see \cite{delmoral2004}, \cite{Douc01} for details.
SMC\ methods are parallelisable and flexible, the latter in the sense that SMC
approximations of the posterior densities for a variety of SSMs can be
constructed quite easily. SMC methods were popularized by the many successful
applications to SSM.

\subsubsection{Path space and fixed-lag approximations}

A SMC\ approximation of $\mathcal{S}_{n}^{\theta}$ may be constructed by
replacing $p_{\theta}\left(  dx_{0:n}|y_{0:n}\right)  $ in Eq.
(\ref{eq:representationadditivefunctional}) with its SMC\ approximation in Eq.
(\ref{eq:filteringdistribution}) - we call this the \emph{path space} method
since the SMC\ approximation of $p_{\theta}\left(  dx_{0:n}|y_{0:n}\right)  $,
which is a probability distribution on $\mathcal{X}^{n+1}$, is used.
Fortunately there is no need to store the entire ancestry of each particle,
i.e. $\left\{  X_{0:n}^{\left(  i\right)  }\right\}  _{1\leq i\leq N}$, which
would require a growing memory. Also, this estimate can be computed
recursively. However, the reliance on the approximation of the joint
distribution $p_{\theta}\left(  \left.  dx_{0:n}\right\vert y_{0:n}\right)  $
is bad. It is well-known in the SMC\ literature that the approximation of
$p_{\theta}\left(  \left.  dx_{0:n}\right\vert y_{0:n}\right)  $\ becomes
progressively impoverished as $n$ increases because of the successive
resampling steps \cite{ADT05}, \cite{delmoraldoucet2003}, \cite{olsson2008}.
That is, the number of distinct samples representing $p_{\theta}\left(
\left.  dx_{0:k}\right\vert y_{0:n}\right)  $\ for any fixed $k<n$ diminishes
as $n$ increases -- this is known as the \emph{particle path degeneracy}
problem. Hence, whatever being the number of particles, $p_{\theta}\left(
\left.  dx_{0:k}\right\vert y_{0:n}\right)  $ will eventually be approximated
by a single unique particle for all (sufficiently large) $n$. This has severe
consequences for the SMC\ estimate $\mathcal{S}_{n}^{\theta}$. In
\cite{delmoraldoucet2003}, under favourable mixing assumptions, the authors
established an upper bound on the $\mathbb{L}^{p}$ error which is proportional
to $n^{2}/\sqrt{N}$. Under similar assumptions, it was shown in
\cite{poyadjis2009} that the asymptotic variance of this
estimate\ increases\ at least quadratically with $n$. To reduce the variance,
alternative methods have been proposed. The technique proposed in
\cite{kitagawa2001} relies on the fact that for a SSM with \textquotedblleft
good\textquotedblright\ forgetting properties,
\begin{equation}
p_{\theta}\left(  \left.  x_{0:k}\right\vert y_{0:n}\right)  \approx
p_{\theta}\left(  \left.  x_{0:k}\right\vert y_{0:\min(k+\Delta,n)}\right)
\label{eq:forgetting}%
\end{equation}
when the horizon $\Delta$ is large enough; that is observations collected
after times $k+\Delta$ bring little additional information about $X_{0:k}$.
(See \cite[Corollary 2.9]{DeG01} for exponential error bounds.) This suggests
that a very simple scheme to curb particle degeneracy is to stop updating the
SMC\ estimate beyond time $k+\Delta$. This algorithm is trivial to implement
but the main practical problem is that of determining an appropriate value for
$\Delta$\ such that the two densities in Eq. (\ref{eq:forgetting}) are close
enough and particle degeneracy is low. These are conflicting requirements. A
too small value for the horizon will result in $p_{\theta}\left(  \left.
x_{0:k}\right\vert y_{0:\min(k+\Delta,n)}\right)  $ being a poor approximation
of $p_{\theta}\left(  \left.  x_{0:k}\right\vert y_{0:n}\right)  $ but the
particle degeneracy will be low. On the other hand, a larger horizon improves
the approximation in Eq. (\ref{eq:forgetting})\ but particle degeneracy will
creep in. Automating the selection of $\Delta$ is difficult. Additionally, for
any finite $\Delta$ the SMC\ estimate of $\mathcal{S}_{n}^{\theta}$ will
suffer from a non vanishing bias even as $N\rightarrow\infty$. In
\cite{olsson2008}, for an optimized value of $\Delta$ which is dependent on
$n$ and the typically unknown mixing properties of the model, the
SMC\ estimates of $\mathcal{S}_{n}^{\theta}$ based on the approximation in Eq.
(\ref{eq:forgetting}) were shown to have an $\mathbb{L}^{p}$ error and bias
upper bounded by quantities proportional to $n\log n/\sqrt{N}$ and $n\log n/N$
under regularity assumptions.

The computational cost of the SMC\ approximation of $\mathcal{S}_{n}^{\theta}%
$\ computed using either the path space method or the truncated horizon method
of \cite{kitagawa2001} is $\mathcal{O}\left(  N\right)  $.

\subsubsection{Approximating the smoothing equations}

A standard alternative to computing $\mathcal{S}_{n}^{\theta}$ is to use
SMC\ approximations of fixed-interval smoothing techniques such as the Forward
Filtering Backward Smoothing (FFBS) algorithm \cite{Douc00}, \cite{godsill}.
Theoretical results on the SMC\ approximations of the FFBS\ algorithm have
been recently established in \cite{douc09}; this includes a central limit
theorem and exponential deviation inequalities. In particular, under
appropriate mixing assumptions, the authors have obtained time-uniform
deviation inequalities for the SMC-FFBS\ approximations of the marginals
$\left\{  p_{\theta}\left(  dx_{k}|y_{0:n}\right)  \right\}  _{0\leq k\leq n}$
\cite[Section 5]{douc09}; see \cite{delmoral2009} for alternative proofs and
complementary results. Let $\widehat{\mathcal{S}}_{n}^{\theta}$ denote the
SMC-FFBS\ estimate of $\mathcal{S}_{n}^{\theta}$. In this work it is
established that the asymptotic variance of $\sqrt{N}\left(  \widehat
{\mathcal{S}}_{n}^{\theta}-\mathcal{S}_{n}^{\theta}\right)  $ only grows
linearly with time $n$; a fact which was also alluded to in
\cite{delmoral2009}. The main advantage of the SMC\ implementation of the
FFBS\ algorithm is that it does not have any tuning parameter other than the
number of particles $N$. However, the improved theoretical properties comes at
a computational price; this algorithm has a computational complexity of
$\mathcal{O}\left(  N^{2}\right)  $ compared to $\mathcal{O}\left(  N\right)
$ for the methods previously discussed. (It is possible to use fast
computational methods to reduce the computational cost to $\mathcal{O}\left(
N\log N\right)  $ \cite{Klass2005}.) Another restriction is that the SMC
implementation of the FFBS\ algorithm does not yield an online algorithm.

\subsection{Contributions and organization of the article}

The contributions of this article are as follows.

\begin{itemize}
\item We propose an original online\ implementation of the SMC-FFBS estimate
of $\mathcal{S}_{n}^{\theta}$. A\ particular case of this new algorithm was
presented in \cite{poyadjis2006}, \cite{poyadjis2009} to compute the score
vector (\ref{eq:score}). However, because it was catered to estimating the
score, the authors failed to realise its full generality.

\item An upper bound for the \emph{non-asymptotic} mean square error of the
SMC-FFBS\ estimate $\widehat{\mathcal{S}}_{n}^{\theta}$ of $\mathcal{S}%
_{n}^{\theta}$ is derived under regularity assumptions. It follows from this
bound that the asymptotic variance of $\sqrt{N}\left(  \widehat{\mathcal{S}%
}_{n}^{\theta}-\mathcal{S}_{n}^{\theta}\right)  $ is bounded by a quantity
proportional to $n$. This complements results recently obtained in
\cite{delmoral2009}, \cite{douc09}.

\item We demonstrate how the online implementation\ of the SMC-FFBS estimate
of $\mathcal{S}_{n}^{\theta}$ can be applied to the problem of recursively
estimating the parameters of a SSM from data. We present original
SMC\ implementations of Recursive Maximum Likelihood (RML) \cite{legland1997},
\cite{poyadjis2006}, \cite{titterington1984} and of the online EM\ algorithm
\cite{ford98}, \cite{elliott2000}, \cite{elliott2002}, \cite{mongillo2008},
\cite{Cap09}, \cite[Section 3.2.]{kantas09}. These SMC implementations do not
suffer from the particle path degeneracy problem.
\end{itemize}

The remainder of this paper is organized as follows. In Section
\ref{sec:smoothedadditivefunctionals} the standard FFBS\ recursion and its
SMC\ implementation is presented. It is then shown how this recursion and its
SMC\ implementation can be implemented exactly with only a forward pass. A
non-asymptotic variance bound is presented in Section \ref{sec:theory}.
Recursive parameter estimation procedures are presented in Section
\ref{sec:experiments}\ and numerical results are given in Section
\ref{sec:simulations}. We conclude in Section \ref{sec:discussion} and the
proof of the main theoretical result is given in the Appendix.

\section{Forward smoothing and SMC\ approximations
\label{sec:smoothedadditivefunctionals}}

We first review the standard FFBS\ recursion and its SMC\ approximation
\cite{Douc00}, \cite{godsill}. This is then followed by a derivation of a
forward-only\ version of the FFBS\ recursion and its corresponding
SMC\ implementation. The algorithms presented in this section do not depend on
any specific SMC\ implementation to approximate $\left\{  p_{\theta}\left(
dx_{n}|y_{0:n}\right)  \right\}  _{n\geq0}$.

\subsection{The forward filtering backward smoothing recursion
\label{sec:FFBS}}

Recall the definition of $\mathcal{S}_{n}^{\theta}$\ in Eq.
(\ref{eq:representationadditivefunctional}). The standard FFBS procedure to
compute $\mathcal{S}_{n}^{\theta}$ proceeds in two steps. In the first step,
which is the forward pass, the filtering densities $\{p_{\theta}\left(
\left.  x_{k}\right\vert y_{0:k}\right)  \}_{0\leq k\leq n}$ are computed
using Bayes' formula:%
\[
p_{\theta}\left(  \left.  x_{k+1}\right\vert y_{0:k+1}\right)  =\frac
{g_{\theta}\left(  \left.  y_{k+1}\right\vert x_{k+1}\right)  \int f_{\theta
}\left(  \left.  x_{k+1}\right\vert x_{k}\right)  p_{\theta}\left(  \left.
x_{k}\right\vert y_{0:k}\right)  dx_{k}}{\int g_{\theta}\left(  \left.
y_{k+1}\right\vert x_{k+1}^{\prime}\right)  f_{\theta}\left(  \left.
x_{k+1}^{\prime}\right\vert x_{k}^{\prime}\right)  p_{\theta}\left(  \left.
x_{k}^{\prime}\right\vert y_{0:k}\right)  dx_{k:k+1}^{\prime}}.
\]
The second step is the backward pass that computes the following marginal
smoothed densities which are needed to evaluate each term in the sum that
defines $\mathcal{S}_{n}^{\theta}$:%
\begin{equation}
p_{\theta}\left(  \left.  x_{k-1},x_{k}\right\vert y_{0:n}\right)  =p_{\theta
}\left(  \left.  x_{k}\right\vert y_{0:n}\right)  p_{\theta}\left(  \left.
x_{k-1}\right\vert y_{0:k-1},x_{k}\right)  ,\quad1\leq k\leq n.
\label{eq:standardforwardbackward}%
\end{equation}
where
\begin{equation}
p_{\theta}\left(  \left.  x_{k-1}\right\vert y_{0:k-1},x_{k}\right)
=\frac{f_{\theta}\left(  \left.  x_{k}\right\vert x_{k-1}\right)  p_{\theta
}\left(  \left.  x_{k-1}\right\vert y_{0:k-1}\right)  }{p_{\theta}\left(
\left.  x_{k}\right\vert y_{0:k-1}\right)  }. \label{eq:conditionaldensity}%
\end{equation}
We compute Eq. (\ref{eq:standardforwardbackward}) commencing at $k=n$ and
then, decrementing $k$ each time, until $k=1$. (Integrating Eq.
(\ref{eq:standardforwardbackward})\ w.r.t. $x_{k}$ will yield $p_{\theta
}\left(  \left.  x_{k-1}\right\vert y_{0:n}\right)  $ which is needed for the
next computation.) To compute $\mathcal{S}_{n}^{\theta}$, $n$ backward steps
must be executed and then $n$ expectations computed. This must then be
repeated at time $n+1$ to incorporate the effect of the new observation
$y_{n+1}$ on these calculations. Clearly this is not an online procedure for
computing $\{\mathcal{S}_{n}^{\theta}\}_{n\geq1}$.

The SMC\ implementation of the FFBS\ recursion is straightforward
\cite{Douc00}. In the forward pass, we compute and store the SMC approximation
$\widehat{p}_{\theta}\left(  \left.  dx_{k}\right\vert y_{0:k}\right)  $ of
$p_{\theta}\left(  \left.  dx_{k}\right\vert y_{0:k}\right)  $ for
$k=0,1,\ldots,n$. In the backward pass, we simply substitute this
SMC\ approximation in the place of $p_{\theta}\left(  \left.  dx_{k}%
\right\vert y_{0:k}\right)  $\ in Eq. (\ref{eq:standardforwardbackward}). Let
\begin{equation}
\widehat{p}_{\theta}\left(  \left.  dx_{k}\right\vert y_{0:n}\right)
=\sum_{i=1}^{N}W_{k|n}^{\left(  i\right)  }\delta_{X_{k}^{\left(  i\right)  }%
}\left(  dx_{k}\right)  \label{eq:smoothingapproximationsFFBS}%
\end{equation}
be the SMC approximation of $p_{\theta}\left(  \left.  dx_{k}\right\vert
y_{0:n}\right)  $, $k\leq n$, initialised at $k=n$ by setting $W_{n|n}%
^{(i)}=W_{n}^{(i)}$. By substituting $\widehat{p}_{\theta}\left(  \left.
dx_{k-1}\right\vert y_{0:k-1}\right)  $ for $p_{\theta}\left(  \left.
x_{k-1}\right\vert y_{0:k-1}\right)  $ in Eq. (\ref{eq:conditionaldensity}),
we obtain
\begin{equation}
\widehat{p}_{\theta}\left(  \left.  dx_{k-1}\right\vert y_{0:k-1}%
,x_{k}\right)  =\frac{\sum_{i=1}^{N}W_{k-1}^{(i)}f_{\theta}\left(
x_{k}|X_{k-1}^{(i)}\right)  \delta_{X_{k-1}^{(i)}}\left(  dx_{k-1}\right)
}{\sum_{l=1}^{N}W_{k-1}^{(l)}f_{\theta}\left(  x_{k}|X_{k-1}^{(l)}\right)  }.
\label{eq:MCconditionaldensity}%
\end{equation}
This approximation is combined with $\widehat{p}_{\theta}\left(  \left.
dx_{k}\right\vert y_{0:n}\right)  $\ (see Eq.
(\ref{eq:standardforwardbackward})) to obtain
\begin{equation}
\widehat{p}_{\theta}(dx_{k-1:k}|y_{0:n})=\sum_{i=1}^{N}\sum_{j=1}^{N}%
W_{k|n}^{(j)}\frac{W_{k-1}^{(i)}f_{\theta}\left(  X_{k}^{(j)}|X_{k-1}%
^{(i)}\right)  }{\sum_{l=1}^{N}W_{k-1}^{(l)}f_{\theta}\left(  X_{k}%
^{(j)}|X_{k-1}^{(l)}\right)  }\delta_{X_{k-1}^{(i)},X_{k}^{\left(  j\right)
}}(dx_{k-1:k}). \label{eq:SMCapproximationsmoothingpairwise}%
\end{equation}
Marginalising this approximation will give the approximation\ to $\widehat
{p}_{\theta}\left(  \left.  dx_{k-1}\right\vert y_{0:n}\right)  $, that is
$\left\{  W_{k-1|n}^{\left(  i\right)  },X_{k-1}^{\left(  i\right)  }\right\}
_{1\leq i\leq N}$, where
\begin{equation}
W_{k-1|n}^{(i)}=\sum_{j=1}^{N}W_{k|n}^{(j)}\frac{W_{k-1}^{(i)}f_{\theta
}\left(  X_{k}^{(j)}|X_{k-1}^{(i)}\right)  }{\sum_{l=1}^{N}W_{k-1}%
^{(l)}f_{\theta}\left(  X_{k}^{(j)}|X_{k-1}^{(l)}\right)  }.
\label{eq:backwardweights}%
\end{equation}
The SMC\ estimate $\widehat{\mathcal{S}}_{n}^{\theta}$ of $\mathcal{S}%
_{n}^{\theta}$ is then given by
\begin{equation}
\widehat{\mathcal{S}}_{n}^{\theta}=\sum\limits_{k=1}^{n}\int s_{k}\left(
x_{k-1},x_{k}\right)  \widehat{p}_{\theta}(dx_{k-1:k}|y_{0:n}).
\label{eq:batchSMCestimate}%
\end{equation}
The backward recursion for the weights, given in Eq. (\ref{eq:backwardweights}%
), makes this an off-line algorithm for computing $\widehat{\mathcal{S}}%
_{n}^{\theta}$.

\subsection{A forward only version of the forward filtering backward smoothing
recursion\label{sec:forwardonlyFFBS}}

To circumvent the need for the backward pass in the computation of
$\mathcal{S}_{n}^{\theta}$, the following auxiliary function (on $\mathcal{X}%
$) is introduced,
\begin{equation}
T_{n}^{\theta}\left(  x_{n}\right)  :=\int S_{n}\left(  x_{0:n}\right)  \text{
}p_{\theta}\left(  \left.  x_{0:n-1}\right\vert y_{0:n-1},x_{n}\right)
dx_{0:n-1}. \label{eq:definitionTn}%
\end{equation}
It is apparent that
\begin{equation}
\mathcal{S}_{n}^{\theta}=\int T_{n}^{\theta}\left(  x_{n}\right)  p_{\theta
}\left(  \left.  x_{n}\right\vert y_{0:n}\right)  dx_{n}.
\label{eq:additivesmoothfunctionalsasfunctionofT}%
\end{equation}
The following proposition establishes a forward recursion to compute
$\{T_{n}^{\theta}\}_{n\geq0}$, which is henceforth referred to as the
\emph{forward smoothing recursion}. For sake of completeness, the proof of
this proposition is given.

\begin{prop}
\label{propDN} For $n\geq1$, we have%
\begin{equation}
T_{n}^{\theta}\left(  x_{n}\right)  =\int\left[  T_{n-1}^{\theta}\left(
x_{n-1}\right)  +s_{n}\left(  x_{n-1},x_{n}\right)  \right]  p_{\theta}\left(
\left.  x_{n-1}\right\vert y_{0:n-1},x_{n}\right)  dx_{n-1},
\label{eq:recursionadditivefunctional}%
\end{equation}
where $T_{0}^{\theta}\left(  x_{0}\right)  :=0$.
\end{prop}

\textbf{Proof.} The proof is straightforward
\begin{align*}
T_{n}^{\theta}\left(  x_{n}\right)   &  :=\int\left[  S_{n-1}\left(
x_{0:n-1}\right)  +s_{n}\left(  x_{n-1},x_{n}\right)  \right]  \text{
}p_{\theta}\left(  \left.  x_{0:n-1}\right\vert y_{0:n-1},x_{n}\right)
dx_{0:n-1}\\
&  =\int\left[  \int S_{n-1}\left(  x_{0:n-1}\right)  p_{\theta}\left(
\left.  x_{0:n-2}\right\vert y_{0:n-2},x_{n-1}\right)  dx_{0:n-2}\right]
p_{\theta}\left(  \left.  x_{n-1}\right\vert y_{0:n-1},x_{n}\right)
dx_{n-1}\\
&  +\int s_{n}\left(  x_{n-1},x_{n}\right)  \text{ }p_{\theta}\left(  \left.
x_{n-1}\right\vert y_{0:n-1},x_{n}\right)  dx_{n-1}.
\end{align*}
The integrand in the first equality is $S_{n}\left(  x_{0:n}\right)  $\ while
the integrand in the first integral of the second equality is $T_{n-1}%
^{\theta}\left(  x_{n-1}\right)  $. $\blacksquare$

This recursion is not new and is actually a special instance of dynamic
programming for Markov processes; see for example \cite{bertsekas78}. For a
fully observed Markov process with transition density $\left\{  f_{\theta
}\left(  \left.  x_{k}\right\vert x_{k-1}\right)  \right\}  _{k\geq1}$, the
dynamic programming recursion to compute the expectation of $S_{n}\left(
x_{0:n}\right)  $ with respect to the law of the Markov process is usually
implemented using a backward recursion going from time $n$ to time $0$. In the
partially observed scenario considered here, $\left\{  X_{k}\right\}  _{0\leq
k\leq n}$ conditional on $y_{0:n}$ is a \textquotedblleft
backward\textquotedblright\ Markov process with non-homogeneous transition
densities $\left\{  p_{\theta}\left(  \left.  x_{k-1}\right\vert
y_{0:k-1},x_{k}\right)  \right\}  _{1\leq k\leq n}$. Thus
(\ref{eq:recursionadditivefunctional}) is the corresponding dynamic
programming recursion to compute $\mathcal{S}_{n}^{\theta}$ with respect to
$p_{\theta}\left(  \left.  x_{0:n}\right\vert y_{0:n}\right)  $ for this
backward Markov chain. This recursion is the foundation of the online EM
algorithm and is described at length in \cite{Elliott1996} (pioneered in
\cite{zeitouni88}) where the density $p_{\theta}\left(  \left.  x_{n-1}%
\right\vert y_{0:n-1},x_{n}\right)  $ appearing in $T_{n}^{\theta}\left(
x_{n}\right)  $ is usually written as
\[
p_{\theta}\left(  \left.  x_{n-1}\right\vert y_{0:n-1},x_{n}\right)
=\frac{f_{\theta}\left(  x_{n}|x_{n-1}\right)  p_{\theta}\left(
x_{n-1},y_{0:n-1}\right)  }{\int f_{\theta}\left(  x_{n}|x_{n-1}\right)
p_{\theta}\left(  x_{n-1},y_{0:n-1}\right)  dx_{n-1}}%
\]
or as in Eq. (\ref{eq:conditionaldensity}) in \cite{Cap09}, \cite{Cap09B},
\cite{delmoral2009}. The forward smoothing recursion has been rediscovered
independently several times; see \cite{hernando05}, \cite{mongillo2008} for example.

A simple SMC\ scheme to approximate $\mathcal{S}_{n}^{\theta}$\ can be devised
by exploiting equations (\ref{eq:additivesmoothfunctionalsasfunctionofT}) and
(\ref{eq:recursionadditivefunctional}). This is summarised as Algorithm SMC-FS below.%

\noindent\hrulefill

\begin{center}
\textbf{Algorithm SMC-FS: Forward-only SMC\ computation of the FFBS estimate}
\end{center}

$\bullet$ \textsf{Assume at time }$n-1$\textsf{ that SMC approximations
}$\left\{  W_{n-1}^{\left(  i\right)  },X_{n-1}^{\left(  i\right)  }\right\}
_{1\leq i\leq N}$\textsf{\ of }$p_{\theta}\left(  \left.  dx_{n-1}\right\vert
y_{0:n-1}\right)  $\textsf{\ and }$\left\{  \widehat{T}_{n-1}^{\theta}\left(
X_{n-1}^{\left(  i\right)  }\right)  \right\}  _{1\leq i\leq N}$\textsf{\ of
}$\left\{  T_{n-1}^{\theta}\left(  X_{n-1}^{\left(  i\right)  }\right)
\right\}  _{1\leq i\leq N}$\textsf{\ are available.}

$\bullet$ \textsf{At time }$n$\textsf{, compute the SMC approximation
}$\left\{  W_{n}^{\left(  i\right)  },X_{n}^{\left(  i\right)  }\right\}
_{1\leq i\leq N}$\textsf{\ of }$p_{\theta}\left(  \left.  dx_{n}\right\vert
y_{0:n}\right)  $\textsf{\ and set}%
\begin{align}
\widehat{T}_{n}^{\theta}\left(  X_{n}^{\left(  i\right)  }\right)   &
=\frac{\sum_{j=1}^{N}W_{n-1}^{\left(  j\right)  }f_{\theta}\left(  X_{n}%
^{(i)}|X_{n-1}^{(j)}\right)  \left[  \widehat{T}_{n-1}^{\theta}\left(
X_{n-1}^{\left(  j\right)  }\right)  +s_{n}\left(  X_{n-1}^{\left(  j\right)
},X_{n}^{\left(  i\right)  }\right)  \right]  }{\sum_{j=1}^{N}W_{n-1}^{\left(
j\right)  }f_{\theta}\left(  X_{n}^{(i)}|X_{n-1}^{(j)}\right)  },\quad1\leq
i\leq N,\label{eq:Tapproximation}\\
\widehat{\mathcal{S}}_{n}^{\theta}  &  =\sum_{i=1}^{N}W_{n}^{\left(  i\right)
}\text{ }\widehat{T}_{n}^{\theta}\left(  X_{n}^{\left(  i\right)  }\right)  .
\label{eq:SMCapproxadditivefunctionals}%
\end{align}%
\noindent\hrulefill

This algorithm is initialized by setting $\widehat{T}_{0}^{\theta}\left(
X_{0}^{\left(  i\right)  }\right)  =0$ for $1\leq i\leq N.$ It has a
computational complexity of $\mathcal{O}\left(  N^{2}\right)  $ which can be
reduced by using fast computational methods \cite{Klass2005}.

The rationale for this algorithm is as follows. By using $\widehat{p}_{\theta
}\left(  \left.  dx_{n-1}\right\vert y_{0:n-1},x_{n}\right)  $ defined in Eq.
(\ref{eq:MCconditionaldensity}) in place of $p_{\theta}\left(  \left.
dx_{n-1}\right\vert y_{0:n-1},x_{n}\right)  $ in Eq.
(\ref{eq:recursionadditivefunctional}), we obtain an approximation
$\widehat{T}_{n}^{\theta}\left(  x_{n}\right)  $ of $T_{n}^{\theta}\left(
x_{n}\right)  $ which is computed at the particle locations $\left\{
X_{n}^{\left(  i\right)  }\right\}  _{1\leq i\leq N}$. The approximation of
$\mathcal{S}_{n}^{\theta}$\ in Eq. (\ref{eq:SMCapproxadditivefunctionals}) now
follows from Eq. (\ref{eq:additivesmoothfunctionalsasfunctionofT}) by using
$\widehat{p}_{\theta}\left(  \left.  dx_{n}\right\vert y_{0:n}\right)  $\ in
place of $p_{\theta}\left(  \left.  dx_{n}\right\vert y_{0:n}\right)  $.

It is valid to use the same notation for the estimates in Eq.
(\ref{eq:batchSMCestimate}) and in Eq. (\ref{eq:SMCapproxadditivefunctionals})
as they are indeed the same. The verification of this assertion may be
accomplished by unfolding the recursion in Eq. (\ref{eq:Tapproximation}).

\section{Theoretical results\label{sec:theory}}

In this section, we present a bound on the non-asymptotic mean square error of
the estimate $\widehat{\mathcal{S}}_{n}^{\theta}$ of $\mathcal{S}_{n}^{\theta
}$. For sake of simplicity, the result is established for additive functionals
of the type
\begin{equation}
S_{n}\left(  x_{0:n}\right)  =\sum_{k=0}^{n}s_{k}\left(  x_{k}\right)
\label{eq:additiveFunctionalSimple}%
\end{equation}
where $s_{k}:\mathcal{X}\rightarrow\mathbb{R}$, and when Algorithm SMC-FS is
implemented using the bootstrap particle filter; see \cite{cappe2005},
\cite{Douc01}\ for a definition of this \textquotedblleft
vanilla\textquotedblright\ particle filter. The result can be generalised to
accommodate an auxiliary implementation of the particle filter
\cite{fearnhead1999}, \cite{douc09}, \cite{Pitt99}. Likewise, the conclusion
is also valid for additive functionals of the type in
(\ref{eq:additiveFunctional}); the proof uses similar arguments but is more complicated.

The following regularity condition will be assumed.

\textbf{(A)} There exist constants $0<\rho,\delta<\infty$ such that for all
$x,x^{\prime}\in\mathcal{X}$, $y\in\mathcal{Y}$ and $\theta\in\Theta$,
\[
\rho^{-1}\leq f_{\theta}\left(  x^{\prime}|x\right)  \leq\rho,\quad\delta
^{-1}\leq g_{\theta}\left(  y|x\right)  \leq\delta.
\]
Admittedly, this assumption is restrictive and\ typically holds when
$\mathcal{X}$\ and $\mathcal{Y}$\ are finite or are compact spaces. In
general, quantifying the errors of SMC approximations under weaker assumptions
is possible \cite{douc09}. (More precise but complicated error bounds for the
particle estimate of $\mathcal{S}_{n}^{\theta}$\ are also presented in
\cite{delmoral2009}\ under weaker assumptions.) However, when (A) holds, the
bounds can be greatly simplified to the extent that they can usually be
expressed as linear or quadratic functions of the time horizon $n$. These
simple rates of growth are meaningful as they have also been observed in
numerical studies even in scenarios where Assumption A\ is not satisfied
\cite{poyadjis2009}.

For a function $s:\mathcal{X}\rightarrow\mathbb{R}$, let $\left\Vert
s\right\Vert =\sup_{x\in\mathcal{X}}|s(x)|$. The oscillation of $s$, denoted
$\mbox{osc}(s)$, is defined to be $\sup{\{|s(x)-s(y)|\;;\;x,y\in\mathcal{X}%
\}}$. The\ main result in this section is the following non-asymptotic bound
for the mean square error of the estimate $\widehat{\mathcal{S}}_{n}^{\theta}$
of $\mathcal{S}_{n}^{\theta}$ given in Eq.
(\ref{eq:SMCapproxadditivefunctionals}).

\begin{theo}
\label{nonasymptheo} Assume (A). Consider the additive functional $S_{n}$\ in
(\ref{eq:additiveFunctionalSimple}) with $\left\Vert s_{k}\right\Vert <\infty$
and $\mbox{osc}(s_{k})\leq1$\ for $0\leq k\leq n$. Then, for any $n\geq0$ and
$\theta\in\Theta$,
\begin{equation}
\mathbb{E}\left(  \left\vert \widehat{\mathcal{S}}_{n}^{\theta}-\mathcal{S}%
_{n}^{\theta}\right\vert ^{2}\right)  \leq a~\frac{(n+1)}{N}\left(
1+\sqrt{\frac{n+1}{N}}\right)  ^{2} \label{LLr}%
\end{equation}
where $a$ is a finite constant that is independent of time $n$, $\theta$ and
the particular choice of functions $\{s_{k}\}_{0\leq k\leq n}$.
\end{theo}

The proof is given in the Appendix. It follows that the asymptotic variance
of\linebreak$\sqrt{N}\left(  \widehat{\mathcal{S}}_{n}^{\theta}-\mathcal{S}%
_{n}^{\theta}\right)  $, i.e. as the number of particles $N$ goes to infinity,
is upper bounded by a quantity proportional to $(n+1)$ as the bias of the
estimate is $\mathcal{O}(1/N)$ \cite[Corollary 5.3]{delmoral2009}.

Let $\widehat{\mathcal{R}}_{n}^{\theta}$\ denote the SMC estimate of
$\mathcal{S}_{n}^{\theta}$ obtained with the standard path space method. This
estimate can have a much larger asymptotic variance as is illustrated with the
following very simple model. Let $f_{\theta}\left(  x^{\prime}|x\right)
=\mu_{\theta}\left(  x^{\prime}\right)  $, i.e. $\{X_{n}\}_{n\geq0}$ is an
i.i.d. sequence, and let $y_{k}=y$ and $s_{k}=s$ for all $k\geq1$ where $s$ is
some real valued function on $\mathcal{X}$, and $s_{0}=0$. It can be easily
established that the formula for the asymptotic variance of $\sqrt{N}\left(
\widehat{\mathcal{R}}_{n}^{\theta}-\mathcal{S}_{n}^{\theta}\right)  $ given in
\cite{chopin2004}, \cite[eqn. (9.13), page 304 ]{delmoral2004} simplifies to
\begin{equation}
n\int\frac{\left[  \pi_{\theta}\left(  \left.  x\right\vert y\right)
\widetilde{s}_{\theta}\left(  x\right)  \right]  ^{2}}{\mu_{\theta}\left(
x\right)  }dx+\frac{n\left(  n-1\right)  }{2}\int\frac{\pi_{\theta}\left(
\left.  x\right\vert y\right)  ^{2}}{\mu_{\theta}\left(  x\right)  }dx\text{
}\int\widetilde{s}_{\theta}\left(  x\right)  ^{2}\pi_{\theta}\left(  \left.
x\right\vert y\right)  dx \label{eq:asymptoticvariancepathbased}%
\end{equation}
where
\begin{align*}
\pi_{\theta}\left(  \left.  x\right\vert y\right)   &  =\frac{\mu_{\theta
}\left(  x\right)  g_{\theta}\left(  \left.  y\right\vert x\right)  }{\int
\mu_{\theta}\left(  x^{\prime}\right)  g_{\theta}\left(  \left.  y\right\vert
x^{\prime}\right)  dx^{\prime}},\\
\widetilde{s}_{\theta}\left(  x\right)   &  =s\left(  x\right)  -\int s\left(
x\right)  \pi_{\theta}\left(  \left.  x\right\vert y\right)  dx.
\end{align*}
Thus the asymptotic variance increases quadratically with time $n$. Note
though that the asymptotic variance of $\sqrt{N}\left(  n^{-1}\widehat
{\mathcal{R}}_{n}^{\theta}-n^{-1}\mathcal{S}_{n}^{\theta}\right)  $ converges
as $n$ tends to infinity to a positive constant. Thus path space method can
provide stable estimates of $\mathbb{E}_{\theta}\left[  \left.  n^{-1}%
S_{n}\left(  X_{0:n}\right)  \right\vert y_{0:n}\right]  $, i.e. when the
additive functionals are time-averaged. Let
\[
S_{\gamma,n}\left(  x_{0:n}\right)  =\gamma_{n}s(x_{n})+\sum_{k=1}%
^{n-1}s\left(  x_{k}\right)  \gamma_{k}\prod\limits_{i=k+1}^{n}(1-\gamma_{i})
\]
where $\left\{  \gamma_{n}\right\}  _{n\geq1}$ is a positive non-increasing
sequence that satisfies the following constraints: $\sum_{n}\gamma_{n}=\infty$
and $\sum_{n}\gamma_{n}^{2}<\infty$. When $\gamma_{n}=n^{-1}$ then
$S_{\gamma,n}\left(  x_{0:n}\right)  =n^{-1}S_{n}\left(  x_{0:n}\right)  $.
One important choice for recursive parameter estimation (see Section
\ref{sec:experiments}) is
\begin{equation}
\gamma_{n}=n^{-\alpha},\quad0.5<\alpha\leq1. \label{eq:stepChoice}%
\end{equation}
It is also of interest to quantify the stability of the path space method when
applied to estimate $\mathcal{S}_{n}^{\theta}=\mathbb{E}_{\theta}\left[
\left.  S_{\gamma,n}\left(  X_{0:n}\right)  \right\vert y_{0:n}\right]  $\ in
this more general time-averaging setting. Once again let $\widehat
{\mathcal{R}}_{n}^{\theta}$\ denote the SMC estimate of $\mathbb{E}_{\theta
}\left[  \left.  S_{\gamma,n}\left(  X_{0:n}\right)  \right\vert
y_{0:n}\right]  $ obtained with the standard path space method. Using the
formula for the asymptotic variance of $\sqrt{N}\left(  \widehat{\mathcal{R}%
}_{n}^{\theta}-\mathcal{S}_{n}^{\theta}\right)  $ given in \cite{chopin2004},
\cite[eqn. (9.13), page 304 ]{delmoral2004} it can be verified that this
asymptotic variance is
\begin{align*}
&  \int\frac{\left[  \pi_{\theta}\left(  \left.  x\right\vert y\right)
\widetilde{s}_{\theta}\left(  x\right)  \right]  ^{2}}{\mu_{\theta}\left(
x\right)  }dx\sum\limits_{k=1}^{n}\gamma_{k}^{2}\prod\limits_{i=k+1}%
^{n}(1-\gamma_{i})^{2}\\
&  +\int\frac{\pi_{\theta}\left(  \left.  x\right\vert y\right)  ^{2}}%
{\mu_{\theta}\left(  x\right)  }dx\text{ }\int\widetilde{s}_{\theta}\left(
x\right)  ^{2}\pi_{\theta}\left(  \left.  x\right\vert y\right)
dx\sum\limits_{k=2}^{n}\sum\limits_{i=1}^{k-1}\gamma_{i}^{2}(1-\gamma
_{i+1})^{2}\cdots(1-\gamma_{n})^{2}%
\end{align*}
It follows from Lemma \ref{lem:stepDiscounting} in Appendix that any
accumulation point of this sequence (in $n$) has to be positive. In contrast,
the asymptotic variance of $\sqrt{N}(\widehat{\mathcal{S}}_{n}^{\theta
}-\mathcal{S}_{n}^{\theta})$, i.e. when $\mathbb{E}_{\theta}\left[  \left.
S_{\gamma,n}\left(  X_{0:n}\right)  \right\vert y_{0:n}\right]  $ is computed
using Algorithm SMC-FS, will converge to zero as $n$ tends to infinity.

\section{Application to SMC\ parameter estimation\label{sec:experiments}}

An important application of the forward smoothing recursion is to parameter
estimation for non-linear non-Gaussian SSMs. We will assume that observations
are generated from an unknown `true' model with parameter value $\theta^{\ast
}\in\Theta$, i.e. $X_{n}|(X_{n-1}=x_{n-1})\sim f_{\theta^{\ast}}(\cdot
|x_{n-1})$ and $Y_{n}|(X_{n}=x_{n})\sim g_{\theta^{\ast}}(\cdot|x_{n})$. The
static parameter estimation problem has generated a lot of interest over the
past decade and many SMC\ techniques have been proposed to solve it; see
\cite{kantas09} for a recent review.

\subsection{Brief literature review}

In a Bayesian approach to the problem, a prior distribution is assigned to
$\theta$ and\ the sequence of posterior densities $\{p\left(  \theta
,x_{0:n}|y_{0:n}\right)  \}_{n\geq0}$ is estimated recursively using
SMC\ algorithms combined with Markov chain Monte Carlo (MCMC) steps
\cite{Andrieu99}, \cite{Fea02}, \cite{Sto02}. Unfortunately these methods
suffer from the particle path degeneracy problem and will result in unreliable
estimates of the model parameters; see \cite{ADT05}, \cite{olsson2008} for a
discussion of this issue. Given a fixed observation record $y_{0:n}$, an
alternative offline MCMC\ approach to estimate $p\left(  \theta,x_{0:n}%
|y_{0:n}\right)  $ has been recently proposed which relies on proposals built
using the SMC\ approximation of $p_{\theta}\left(  \left.  x_{0:n}\right\vert
y_{0:n}\right)  $ \cite{andrieu2009}.

In a ML approach, the estimate of $\theta^{\ast}$ is the maximising argument
of the likelihood of the observed data. The ML estimate can be calculated
using a gradient ascent algorithm either offline for a fixed batch of data or
online \cite{legland1997}; see Section \ref{subsec:rml}. Likewise, the EM
algorithm can also be implemented offline or online. The online EM algorithm,
assuming all calculations can be performed exactly, is presented in
\cite{ford98}, \cite{elliott2000}, \cite{elliott2002}, \cite{mongillo2008} and
\cite{Cap09B}. For a general SSM for which the quantities required by the
online EM cannot be calculated exactly, an SMC implementation is possible
\cite{Cap09}, \cite[Section 3.2.]{kantas09}; see Section
\ref{sec:parameterestimation}.

\subsection{Gradient ascent algorithms\label{subsec:rml}}

To maximise the likelihood $p_{\theta}(y_{0:n})$ w.r.t. $\theta$, we can use a
simple gradient algorithm. Let $\{\theta_{i}\}_{i\in\mathbb{N}}$\ be the
sequence of parameter estimates of the gradient algorithm. We update the
parameter at iteration $i+1$ using%
\[
\theta_{i+1}=\theta_{i}+\gamma_{i+1}\left.  \nabla\log p_{\theta}\left(
y_{0:n}\right)  \right\vert _{\theta=\theta_{i}}%
\]
where $\left.  \nabla\log p_{\theta}\left(  y_{0:n}\right)  \right\vert
_{\theta=\theta_{i}}$ is the score vector computed at $\theta=\theta_{i}$ and
$\left\{  \gamma_{i}\right\}  _{i\geq1}$ is a sequence of positive
non-increasing step-sizes defined in (\ref{eq:stepChoice}). For a general SSM,
we need to approximate $\left.  \nabla\log p_{\theta}\left(  y_{0:n}\right)
\right\vert _{\theta=\theta_{i}}$. As mentioned in the introduction, the score
vector admits several smoothed additive functional representations; see Eq.
(\ref{eq:score}) and \cite{coquelin2008}. Using Eq. (\ref{eq:score}), it is
possible to approximate the score with Algorithm SMC-FS.

In the online implementation, the parameter estimate at time $n+1$ is updated
according to \cite{BMP90}, \cite{legland1997}
\begin{equation}
\theta_{n+1}=\theta_{n}+\gamma_{n+1}\nabla\log p_{\theta_{0:n}}(y_{n}%
|y_{0:n-1}) \label{eq:RML}%
\end{equation}
Upon receiving $y_{n}$, $\theta_{n}$ is updated in the direction of ascent of
the predictive density of this new observation. A necessary requirement for an
online implementation is that the previous values of the model parameter
estimates (other than $\theta_{n}$) are also used in the evaluation of
$\nabla_{\theta}\log p_{\theta}(y_{n}|y_{0:n-1})$ at $\theta=\theta_{n}$. This
is indicated in the notation $\nabla\log p_{\theta_{0:n}}(y_{n}|y_{0:n-1})$.
(Not doing so would require browsing through the entire history of
observations.) This approach was suggested by \cite{legland1997} for the
finite state-space case and is named RML. The asymptotic properties of this
algorithm (i.e. the behavior of $\theta_{n}$\ in the limit as $n$\ goes to
infinity) have been studied in the case of an i.i.d. hidden process by
\cite{titterington1984} and for an HMM with a finite state-space by
\cite{legland1997}. Under suitable regularity assumptions, convergence to
$\theta^{\ast}$ and a central limit theorem for the estimation error has been established.

For a general SSM, we can compute a SMC\ estimate of $\nabla\log
p_{\theta_{0:n}}(y_{n}|y_{0:n-1})$ using Algorithm SMC-FS upon noting that
$\nabla\log p_{\theta_{0:n}}(y_{n}|y_{0:n-1})$ is equal to
\[
\nabla\log p_{\theta_{0:n}}(y_{0:n})-\nabla\log p_{\theta_{0:n-1}}%
(y_{0:n-1}).
\]
In particular, at time $n$, a particle approximation $\left\{  W_{n}^{\left(
i\right)  },X_{n}^{\left(  i\right)  }\right\}  _{1\leq i\leq N}$ of
$p_{\theta_{0:n}}\left(  \left.  dx_{n}\right\vert y_{0:n}\right)  $ is
computed using the particle approximation at time $n-1$ and parameter value
$\theta=\theta_{n}$. Similarly, the computation of Eq.
(\ref{eq:Tapproximation}) is performed using $\theta=\theta_{n}$\ and with
\[
s_{n}(x_{n-1},x_{n})=\left.  \nabla\log f_{\theta}\left(  \left.
x_{n}\right\vert x_{n-1}\right)  \right\vert _{\theta=\theta_{n}}+\left.
\nabla\log g_{\theta}\left(  \left.  y_{n}\right\vert x_{n}\right)
\right\vert _{\theta=\theta_{n}}.
\]
The estimate of $\nabla\log p_{\theta_{0:n}}(y_{n}|y_{0:n-1})$\ is now the
difference of the estimate in Eq. (\ref{eq:SMCapproxadditivefunctionals}) with
the same estimate computed at time $n-1$.

Under the regularity assumptions given in Section \ref{sec:theory}, it follows
from the results in the Appendix that the asymptotic variance (i.e. as
$N\rightarrow\infty$) of the SMC\ estimate of $\nabla\log p_{\theta_{0:n}%
}(y_{n}|y_{0:n-1})$ computed using Algorithm SMC-FS is uniformly (in time)
bounded. On the contrary, the standard path-based SMC\ estimate of $\nabla\log
p_{\theta_{0:n}}(y_{n}|y_{0:n-1})$ has an asymptotic variance that increases
linearly with $n$.

\subsection{EM algorithms\label{sec:parameterestimation}}

Gradient ascent algorithms are more generally applicable than the EM
algorithm. However, their main drawback in practice is that it is difficult to
properly scale the components of the computed gradient vector. For this reason
the EM\ algorithm is usually favoured by practitioners whenever it is applicable.

Let $\{\theta_{i}\}_{i\in\mathbb{N}}$\ be the sequence of parameter estimates
of the EM algorithm. In the offline approach, at iteration $i+1$, the function%
\[
Q(\theta_{i},\theta)=\int\log p_{\theta}(x_{0:n},y_{0:n})\text{ }p_{\theta
_{i}}(x_{0:n}|y_{0:n})dx_{0:n}%
\]
is computed and then maximized. The maximizing argument is the new estimate
$\theta_{i+1}$. If\ $p_{\theta}(x_{0:n},y_{0:n})$ belongs to the exponential
family, then the maximization step is usually straightforward. We now give an
example of this.

Let $s^{l}:\mathcal{X\times X\times Y}\rightarrow\mathbb{R}$, $l=1,\ldots,m$,
be a collection of functions with corresponding additive functionals%
\[
S_{l,n}(x_{0:n},y_{0:n})=\sum_{k=1}^{n}s^{l}\left(  x_{k-1},x_{k}%
,y_{k}\right)  ,\qquad1\leq l\leq m,
\]
and let
\[
\mathcal{S}_{l,n}^{\theta}=\int S_{l,n}(x_{0:n},y_{0:n})p_{\theta}%
(x_{0:n}|y_{0:n})dx_{0:n}.
\]
The collection $\{\mathcal{S}_{l,n}^{\theta}\}_{1\leq l\leq m}$ is also
referred to as the \emph{summary statistics} in the literature.\ Typically,
the maximising argument of $Q(\theta_{i},\theta)$\ can be characterised
explicitly through a suitable function $\Lambda:\mathbb{R}^{m}\rightarrow
\Theta$, i.e.
\begin{equation}
\theta_{i+1}=\Lambda\left(  n^{-1}\mathcal{S}_{n}^{\theta_{i}}\right)
\end{equation}
where $[\mathcal{S}_{n}^{\theta}]_{l}=\mathcal{S}_{l,n}^{\theta}$. As an
example of this, consider the following stochastic volatility model
\cite{Pitt99}.

\begin{exam}
\label{ex:stochVol}The stochastic volatility model is a SSM defined by the
following equations:
\begin{align*}
X_{0}  &  \sim\mathcal{N}\left(  0,\frac{\sigma^{2}}{1-\phi^{2}}\right)
,\text{ }X_{n+1}=\phi X_{n}+\sigma V_{n+1},\\
Y_{n}  &  =\beta\exp\left(  X_{n}/2\right)  W_{n},
\end{align*}
where $\{V_{n}\}_{n\in\mathbb{N}}$ and $\{W_{n}\}_{n\geq0}$ are independent
and identically distributed standard normal noise sequences, which are also
independent of each other and of the initial state $X_{0}$. The model
parameters $\theta\triangleq\left(  \phi,\sigma^{2},\beta^{2}\right)
\in\mathbb{R\times(}0,\infty\mathbb{)}\times\mathbb{(}0,\infty\mathbb{)}$ are
to be estimated. To apply the EM algorithm to this model, let%
\begin{align*}
s^{1}\left(  x_{n-1},x_{n},y_{n}\right)   &  =x_{n-1}x_{n},\text{ }%
s^{2}\left(  x_{n-1},x_{n},y_{n}\right)  =\left(  x_{n-1}\right)  ^{2},\\
s^{3}\left(  x_{n-1},x_{n},y_{n}\right)   &  =\left(  x_{n}\right)
^{2},\text{ }s^{4}\left(  x_{n-1},x_{n},y_{n}\right)  =y_{n}^{2}\exp\left(
-x_{n}\right)  .
\end{align*}
For large $n$, we can safely ignore the terms associated to the initial
density $\mu_{\theta}\left(  x\right)  $ and the solution to the maximisation
step is characterised by the function
\[
\Lambda(z_{1},z_{2},z_{3},z_{4})=\left(  \frac{z_{1}}{z_{2}},z_{3}+\left(
\frac{z_{1}}{z_{2}}\right)  ^{2}z_{2}-2\left(  \frac{z_{1}}{z_{2}}\right)
z_{1},z_{4}\right)  .
\]

\end{exam}

The SMC\ implementation of the forward smoothing recursion has advantages even
for the batch EM algorithm. As there is no backward pass, there is no need to
store the particle approximations of $\left\{  p_{\theta}(dx_{k}%
|y_{0:k})\right\}  _{k=0,...,n}$, which can result in a significant memory
saving for large data sets.

In the online implementation, running averages of the sufficient statistics
are computed instead \cite{Cap09}, \cite{elliott2000}, \cite{elliott2002},
\cite{ford98}, \cite[Section 3.2.]{kantas09}, \cite{mongillo2008}. Let
$\{\theta_{k}\}_{0\leq k\leq n}$ be the sequence of parameter estimates of the
online EM\ algorithm computed sequentially based on $y_{0:n}$. When $y_{n+1}$
is received, for each $l=1,\ldots,m$, compute
\begin{equation}
\hspace{-0.5cm}%
\begin{tabular}
[c]{l}%
$\mathcal{S}_{l,n+1}=\gamma_{n+1}\text{ }\int s^{l}\left(  x_{n}%
,x_{n+1},y_{n+1}\right)  p_{\theta_{0:n}}(x_{n},x_{n+1}|y_{0:n+1})dx_{n:n+1}%
$\\
$\ \ \ \ \ \ \ \ \ \ \ +\left(  1-\gamma_{n+1}\right)  \int\sum_{k=1}%
^{n}\left(  \prod\limits_{i=k+1}^{n}\left(  1-\gamma_{i}\right)  \right)
\gamma_{k}s^{l}\left(  x_{k-1},x_{k},y_{k}\right)  p_{\theta_{0:n}}%
(x_{0:n}|y_{0:n+1})dx_{0:n},$%
\end{tabular}
\ \ \ \ \ \ \label{eq:suffStatOnline}%
\end{equation}
and then set%
\[
\theta_{n+1}=\Lambda\left(  \mathcal{S}_{n+1}\right)
\]
where $[\mathcal{S}_{n+1}]_{l}=\mathcal{S}_{l,n+1}$. Here $\{\gamma
_{n}\}_{n\geq1}$ is a step-size sequence satisfying the same conditions
stipulated for the RML in Section \ref{subsec:rml}. (The recursive
implementation of $\mathcal{S}_{l,n+1}$\ is standard \cite{BMP90}.) The
subscript $\theta_{0:n}\ $on $p_{\theta_{0:n-1}}(x_{0:n}|y_{0:n})$ indicates
that the posterior density is being computed sequentially using the parameter
$\theta_{k-1}$ at time $k$ (and $\theta_{0}$ at time $0$.) References
\cite{Cap09B}, \cite{elliott2000}, \cite[chapter 4]{ford98} and
\cite{mongillo2008} have proposed an online EM algorithm, implemented as
above, for finite state HMMs. In the finite state setting all computations
involved can be done exactly in contrast to general SSMs where numerical
procedures are called for. It is also possible to do all the calculations
exactly for linear Gaussian models \cite{elliott2002}.

Define the vector valued function $s:\mathcal{X\times X\times Y}%
\rightarrow\mathbb{R}^{m}$ as follows: $s=[s^{1},\ldots,s^{m}]^{\text{T}}$.
Computing $\mathcal{S}_{n}$ sequentially using SMC-FS is straightforward and
detailed in the following\ algorithm.%

\noindent\hrulefill
\vspace{-0.25cm}

\begin{center}
\textbf{SMC-FS implementation of online EM}
\end{center}

\textsf{\hspace{-0.5cm}}\underline{\textsf{At time }$n=0$}

$\hspace{-0.5cm}\bullet$\textsf{\ Choose }$\theta_{0}$\textsf{.}

$\hspace{-0.5cm}\bullet$\textsf{\ Set }$T_{0}^{\left(  i\right)  }%
=0\in\mathbb{R}^{m}$\textsf{, }$i=1,\ldots,N$\textsf{.}

$\hspace{-0.5cm}\bullet$\textsf{\ Construct the SMC approximation }%
$\{X_{0}^{\left(  i\right)  },W_{0}^{\left(  i\right)  }\}_{1\leq i\leq N}%
$\textsf{\ of }$p_{\theta_{0}}(dx_{0}|y_{0})$\textsf{.}

\textsf{\hspace{-0.5cm}}\underline{\textsf{At times }$n\geq1$}

$\hspace{-0.5cm}\bullet$\textsf{\ Construct the SMC approximation }%
$\{X_{n}^{\left(  i\right)  },W_{n}^{\left(  i\right)  }\}_{1\leq i\leq N}%
$\textsf{\ of }$p_{\theta_{0:n-1}}(dx_{n}|y_{0:n})$.

$\hspace{-0.5cm}\bullet$\textsf{\ For each }$i=1,\ldots,N$, \textsf{compute }%
\[
T_{n}^{(i)}=\frac{\sum_{j=1}^{N}W_{n-1}^{\left(  j\right)  }f_{\theta_{n-1}%
}\left(  X_{n}^{(i)}|X_{n-1}^{(j)}\right)  \left[  \left(  1-\gamma
_{n}\right)  \text{ }T_{n-1}^{(j)}+\gamma_{n}\text{ }s\left(  X_{n-1}^{\left(
j\right)  },X_{n}^{\left(  i\right)  },y_{n}\right)  \right]  }{\sum_{j=1}%
^{N}W_{n-1}^{\left(  j\right)  }f_{\theta_{n-1}}\left(  X_{n}^{(i)}%
|X_{n-1}^{(j)}\right)  }.
\]

$\hspace{-0.5cm}\bullet$\textsf{\ Compute }$\widehat{\mathcal{S}}_{n}%
=\sum_{i=1}^{N}W_{n}^{\left(  i\right)  }T_{n}^{(i)}$ \textsf{and update the
parameter, }$\theta_{n}=\Lambda\left(  \widehat{\mathcal{S}}_{n}\right)  .$

\vspace{-0.25cm}%
\noindent\hrulefill

It was suggested in \cite[Section 3.2.]{kantas09} that the two other SMC
methods discussed in Section \ref{sec:litReview} could be used to approximate
$\mathcal{S}_{n}$; the path space approach\ to implement the online EM was
also independently proposed in \cite{Cap09}. Doing so would yield a cheaper
alternative to Algorithm SMC-EM above with computational cost $\mathcal{O}%
\left(  N\right)  $, but not without its drawbacks. The fixed-lag
approximation of \cite{kitagawa2001} would introduce a bias which might be
difficult to control and the path space approach suffers from the usual
particle path degeneracy problem. Consider the step-size sequence in
(\ref{eq:stepChoice}). If the path space method is used to
estimate\ $\mathcal{S}_{n}$ then the theory in Section \ref{sec:theory} tells
us that, even under strong mixing assumptions, the asymptotic variance of the
estimate of $\mathcal{S}_{n}$\ will not converge to zero for $0.5<\alpha\leq
1$. Thus it will not yield a theoretically convergent algorithm. Numerical
experiments in \cite{Cap09} appear to provide stable results which we
attribute to the fact that this variance might be very small in the scenarios
considered\footnote{In a\ Bayesian framework where $\theta$ is assigned a
prior distribution and we estimate $\left\{  p\left(  \left.  x_{n}%
,\theta\right\vert y_{0:n}\right)  \right\}  _{n\geq0}$ \cite{Andrieu99},
\cite{Fea02}, \cite{Sto02}, the path degeneracy problem has much more severe
consequences than in the ML framework considered here as illustrated in
\cite{ADT05}. Indeed in the ML framework, the filter $\left\{  p_{\theta
}\left(  \left.  x_{n}\right\vert y_{0:n}\right)  \right\}  _{n\geq0}$ will
have, under regularity assumptions, exponential forgetting properties for any
$\theta\in\Theta$ whereas this will never be the case for $p\left(  \left.
x_{n},\theta\right\vert y_{0:n}\right)  .$}. In contrast, the asymptotic
variance of the $\mathcal{O}\left(  N^{2}\right)  $ estimate\ converges to
zero in time $n$ for the entire range $0.5<\alpha\leq1$ under the same mixing
conditions. The original $\mathcal{O}\left(  N^{2}\right)  $ implementation
proposed here has been recently successfully adopted in \cite{lecorff2010} to
solve a complex parameter estimation problem arising in robotics.

\section{Simulations\label{sec:simulations}}

\subsection{Comparing SMC-FS with the path space method}

We commence with a study of a scalar linear Gaussian SSM for which we may
calculate smoothed functionals analytically. We use these exact values as
benchmarks for the SMC\ approximations. The model is
\begin{align}
X_{0}  &  \sim\mathcal{N}\left(  0,\sigma_{0}^{2}\right)  ,\text{ }%
X_{n+1}=\phi X_{n}+\sigma_{V}V_{n+1},\text{ }\label{eq:evollineargauss}\\
Y_{n}  &  =cX_{n}+\sigma_{W}W_{n}, \label{eq:obslineargauss}%
\end{align}
where $V_{n}\overset{\text{i.i.d.}}{\sim}\mathcal{N}\left(  0,1\right)  $ and
$W_{n}\overset{\text{i.i.d.}}{\sim}\mathcal{N}\left(  0,1\right)  $. We
compared the exact values of the following smoothed functionals%
\begin{equation}
\mathcal{S}_{1,n}^{\theta}=\mathbb{E}_{\theta}\left[  \left.  \sum_{k=1}%
^{n}X_{k-1}^{2}\right\vert y_{0:n}\right]  ,\quad\mathcal{S}_{2,n}^{\theta
}=\mathbb{E}_{\theta}\left[  \left.  \sum_{k=1}^{n}X_{k-1}\right\vert
y_{0:n}\right]  ,\quad\mathcal{S}_{3,n}^{\theta}=\mathbb{E}_{\theta}\left[
\left.  \sum_{k=1}^{n}X_{k-1}X_{k}\right\vert y_{0:n}\right]  ,
\label{eq:benchmarkFunctionals}%
\end{equation}
computed at $\theta^{\ast}=\left(  \phi^{\ast},\sigma_{V}^{\ast},c^{\ast
},\sigma_{W}^{\ast}\right)  =\left(  0.8,0.1,1.0,1.0\right)  $\ with the
bootstrap filter implementation of Algorithm SMC-FS and the path space method.
Comparisons were made after 2500, 5000, 7500 and 10,000 observations to
monitor the increase in variance and the experiment was replicated 50 times to
generate the box-plots in Figure \ref{combinedNandN2boxplots}. (All
replications used the same data record.) Both estimates were computed using
$N=500$ particles.

\begin{figure}[h]
\centering\includegraphics[width=1.0\textwidth]{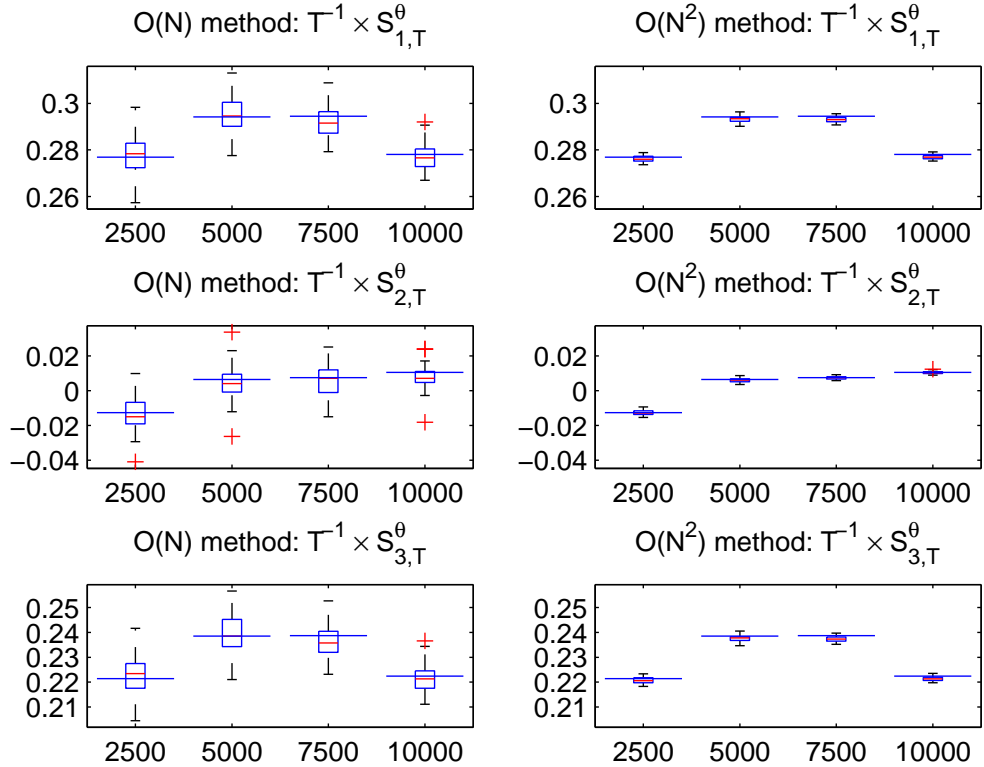}\caption{Box
plots of SMC estimates of the smoothed additive functions in
(\ref{eq:benchmarkFunctionals}) for a linear Gaussian SSM. Estimates were
computed with path space method (left column) and Algorithm SMC-FS (right
column). The long horizontal line intersecting the box indicates the true
value.}%
\label{combinedNandN2boxplots}%
\end{figure}

From Figure \ref{combinedNandN2boxplots}\ it is evident that the
SMC\ estimates of Algorithm SMC-FS significantly\ outperforms the
corresponding SMC\ estimates of the path space method. However one should bear
in mind that the former algorithm has $\mathcal{O}(N^{2})$ computational
complexity while the latter is $\mathcal{O}(N)$. Thus a comparison that takes
this difference into consideration is important. From Theorem
\ref{nonasymptheo}\ and the discussion after it, we expect the variance of
Algorithm SMC-FS's estimate to grow only linearly with the time index compared
to a quadratic in time growth of variance for the path space method. Hence,
for the same computational effort we argue that, for large observation
records, the estimate of Algorithm SMC-FS is always going to outperform the
path space estimates. Specifically, for a large enough $n$, the variance of
Algorithm SMC-FS's estimate with $N$ particles will be significantly less than
the variance of the path space estimate with $N^{2}$ particles. If the number
of observations is small then, taking into account the computational
complexity, it might be better to use the path space estimate as the variance
benefit of using Algorithm SMC-FS may not be appreciable to justify the
increased computational load.

\subsection{Online EM}

Figure \ref{fig:constantStep} shows the parameter estimates obtained using the
SMC implementation of online EM for the stochastic volatility model discussed
in Example \ref{ex:stochVol}. The true value of the parameters were
$\theta^{\ast}=\left(  \phi,\sigma^{2},\beta^{2}\right)  =\left(
0.8,0.1,1\right)  $\ and 500 particles were used. SMC-EM\ was started at the
initial guess $\theta_{0}=(0.1,1,2)$. For the first 100 observations, only the
E-step was executed. That is the step $\theta_{n}=\Lambda\left(
\widehat{\mathcal{S}}_{n}\right)  $,\ which is the M-step was skipped.
SMC-EM\ was run in its entirety for observations 101 and onwards. The step
size used was $\gamma_{n}=0.01$ for $n\leq10^{5}$ and $1/(n-\,\,5\times
10^{4})^{0.6}$ for $n>10^{5}$. Figure \ref{fig:constantStep} shows the
sequence of parameter estimates computed with a very long\ observation sequence.

\begin{figure}[h]
\centering\includegraphics[width=0.6\textwidth]{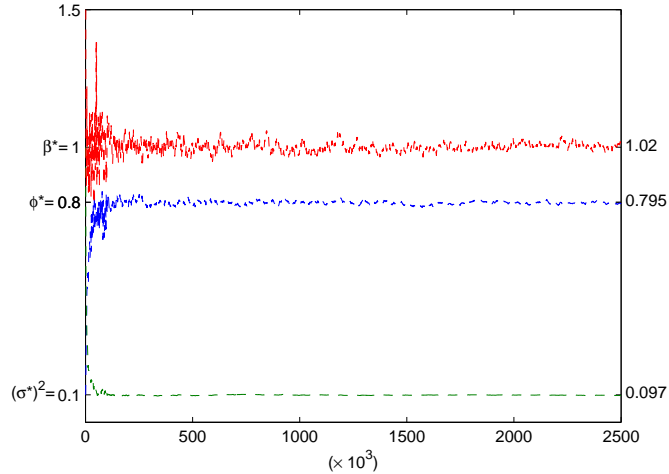}\caption{
Estimating the parameters of the stochastic volatility model with the SMC
version of online EM, Algorithm SMC-EM. Initial parameter guess $\theta
_{0}=(0.1,1,2)$. True and converged values (average of the last 1000
iterations) are indicated on the left and the right of the plot respectively.}%
\label{fig:constantStep}%
\end{figure}


\section{Discussion\label{sec:discussion}}

We proposed a new SMC algorithm to compute the expectation of additive
functionals recursively in time. Essentially, it is an online\ implementation
of the FFBS SMC algorithm proposed in \cite{Douc00}. This algorithm has an
$\mathcal{O}(N^{2})$ computational complexity where $N$ is the number of
particles. It was mentioned how a standard path space SMC\ estimator to
compute the same expectations recursively in time could be developed. This
would have an $\mathcal{O}(N)$ computational complexity. However, as
conjectured in \cite{poyadjis2009}, it was shown here that the asymptotic
variance of the SMC-FFBS estimator\ increased linearly with time whereas that
of the $\mathcal{O}(N)$\ method increased quadratically. The online SMC-FFBS
estimator was then used to\ perform recursive parameter estimation. While the
convergence of RML and online EM have been established when they can be
implemented exactly, the convergence of the SMC\ implementation of these
algorithms have yet to be established and is currently under investigation.

\section{Acknowledgments}

The authors would like to thank Olivier Capp\'{e}, Thomas Flury, Sinan
Yildirim and \'{E}ric Moulines for comments and references that helped improve
the first version of this paper. The authors are also grateful to Rong Chen
for pointing out the link between the forward smoothing recursion and dynamic
programming. Finally, we are thankful to Robert Elliott to have pointed out to
us references \cite{elliott2000}, \cite{elliott2002} and \cite{ford98}.

\appendix

\section{Appendix}

The proofs in this section hold for any fixed $\theta$ and therefore $\theta
$\ is omitted from the notation. This section commences with some essential definitions.

Consider the measurable space $(E,\mathcal{E})$.\ Let $\mathcal{M}(E)$ denote
the set of all finite signed measures\ and $\mathcal{P}(E)$ the set of all
probability measures on $E$. Let $\mathcal{B}(E)$\ denote the Banach space of
all bounded and measurable functions $f$ equipped with the uniform norm $\Vert
f\Vert$. Let $\nu(f)=\int~\nu(dx)~f(x)$, i.e. $\nu(f)$ is the Lebesgue
integral of the function $f\in\mathcal{B}(E)$ w.r.t. the measure $\nu
\in\mathcal{M}(E)$. If $\nu$\ is a density w.r.t. some dominating measure $dx$
on $E$ then, $\nu(f)=\int\nu(x)~f(x)dx$. We recall that a bounded integral
kernel $M(x,dx^{\prime})$ from a measurable space $(E,\mathcal{E})$ into an
auxiliary measurable space $(E^{\prime},\mathcal{E}^{\prime})$ is an operator
$f\mapsto M(f)$ from $\mathcal{B}(E^{\prime})$ into $\mathcal{B}(E)$ such that
the functions
\[
x\mapsto M(f)(x):=\int_{E^{\prime}}M(x,dx^{\prime})f(x^{\prime})
\]
are $\mathcal{E}$-measurable and bounded, for any $f\in\mathcal{B}(E^{\prime
})$. In the above displayed formulae, $dx^{\prime}$ stands for an
infinitesimal neighborhood of a point $x^{\prime}$ in $E^{\prime}$. Let
$\beta(M)$ denote the Dobrushin coefficient of $M$ which defined by the
following formula
\[
\beta(M):=\sup{\ \{\mbox{\rm osc}(M(f))\;;\;\;f\in\mbox{\rm Osc}_{1}%
(E^{\prime})\}}%
\]
where $\mbox{Osc}_{1}(E^{\prime})$ stands the set of $\mathcal{E^{\prime}}%
$-measurable functions $f$ with oscillation less than or equal to 1. The
kernel $M$ also generates a dual operator $\nu\mapsto\nu M$ from
$\mathcal{M}(E)$ into $\mathcal{M}(E^{\prime})$ defined by $(\nu
M)(f):=\nu(M(f))$. A Markov kernel is a positive and bounded integral operator
$M$ with $M(1)=1$. Given a pair of bounded integral operators $(M_{1},M_{2})$,
we let $(M_{1}M_{2})$ the composition operator defined by $(M_{1}%
M_{2})(f)=M_{1}(M_{2}(f))$. For time homogenous state spaces, we denote by
$M^{m}\,\left(  =M^{m-1}M=MM^{m-1}\right)  $ the $m$-th composition of a given
bounded integral operator $M$, with $m\geq1$.

Given a positive function $G$ on $E$, let $\Psi_{G}~:~\nu\in\mathcal{P}%
(E)\mapsto\Psi_{G}(\nu)\in\mathcal{P}(E)$ be the Bayes transformation defined
by
\[
\Psi_{G}(\nu)(dx):=\frac{1}{\nu(G)}~G(x)~\nu(dx)
\]
The definitions above also apply if $\nu$\ is a density and $M$ is a
transition density. In this case all instances of $\nu(dx)$ should be replaced
with $\nu(x)dx$\ and $M(x,dx^{\prime})$ by $M(x,x^{\prime})dx^{\prime}$ where
$dx$ and $dx^{\prime}$ are the dominating measures.

The proofs below will apply to any fixed sequence of observation
$\{y_{n}\}_{n\geq0}$ and it is convenient to introduce the following
transition kernels,
\begin{align*}
Q_{n}(x_{n-1},dx_{n})  &  =g(y_{n-1}|x_{n-1})f(x_{n}|x_{n-1})dx_{n},\quad
n\geq1,\\
\mbox{\rm and}\quad Q_{k,n}  &  =Q_{k+1}Q_{k+2}\ldots Q_{n},\quad0\leq k\leq
n,
\end{align*}
with the convention that $Q_{n,n}=Id$, the identity operator. Note that
$Q_{k,n}(1)=p(y_{k:n-1}|x_{k})$.\ 

Let the mapping $\Phi_{k}:\mathcal{P}(\mathcal{X})\rightarrow\mathcal{P}%
(\mathcal{X})$, $k\geq1$, be defined as follows
\[
\Phi_{k}(\nu)(dx_{k})=\frac{\nu Q_{k}(dx_{k})}{\nu Q_{k}(1)}.
\]
Several probability densities and their SMC\ approximations are introduced to
simplify the exposition.\ The \emph{predicted filter} is denoted by
\[
\eta_{n}(dx_{n})=p(dx_{n}\left\vert y_{0:n-1}\right.  )
\]
with the understanding that $\eta_{0}(dx_{0})$ is the initial distribution of
$X_{0}$. Let $\eta_{n}^{N}$\ denote its SMC approximation with $N$ particles.
(This notation for the SMC approximation is opted for, instead of the usual
$\widehat{\eta}_{n}$, to make the number of particles explicit.) The bounded
integral operator $D_{k,n}$ from $\mathcal{X}$ into $\mathcal{X}^{n+1}$ is
defined as
\begin{equation}
D_{k,n}(S_{n})(x_{k}):=\int p(dx_{0:k-1}\left\vert x_{k},y_{0:k-1}\right.
)\left(  \prod\limits_{q=k}^{n-1}g(y_{q}|x_{q})f(x_{q+1}|x_{q})\right)
S_{n}(x_{0:n})dx_{k+1:n} \label{defiDpn}%
\end{equation}
$D_{k,n}$ is defined for any pair of time indices $k,n$ satisfying $0\leq
k\leq n$ with the convention that $p(x_{0:k-1}\left\vert x_{k},y_{0:k-1}%
\right.  )=1$ for $k=0$ and $\prod\emptyset=1$. The SMC\ approximation,
$D_{k,n}^{N}$, is%
\begin{equation}
D_{k,n}^{N}(S_{n})(x_{k}):=\int p^{N}(dx_{0:k-1}\left\vert x_{k}%
,y_{0:k-1}\right.  )\left(  \prod\limits_{q=k}^{n-1}g(y_{q}|x_{q}%
)f(x_{q+1}|x_{q})\right)  S_{n}(x_{0:n})dx_{k+1:n}%
\end{equation}
where $p^{N}(dx_{0:k-1}\left\vert x_{k},y_{0:k-1}\right.  )$ is the
SMC\ approximation of $p(dx_{0:k-1}\left\vert x_{k},y_{0:k-1}\right.  )$
obtained from the SMC-FFBS approximation of Section \ref{sec:FFBS}, i.e.%
\begin{equation}
p^{N}(dx_{0:k-1}\left\vert x_{k},y_{0:k-1}\right.  )=\prod\limits_{q=1}%
^{k}M_{q,\eta_{q-1}^{N}}(x_{q},dx_{q-1})
\end{equation}
where the backward Markov transition kernels $M_{q,\eta_{q-1}^{N}}$\ are
defined through
\begin{equation}
M_{q,\eta_{q-1}^{N}}(x_{q},dx_{q-1})=\frac{\eta_{q-1}^{N}(dx_{q-1}%
)g(y_{q-1}|x_{q-1})f(x_{q}|x_{q-1})}{\eta_{q-1}^{N}(g(y_{q-1}|\cdot
)f(x_{q}|\cdot))}.
\end{equation}
It is easily established that the SMC-FFBS\ approximation of $p(dx_{k}%
|y_{0:n})$, $k\leq n$, is precisely the marginal of
\[
p^{N}(dx_{0:n-1}\left\vert x_{n},y_{0:n-1}\right.  )\widehat{p}(dx_{n}%
|y_{0:n})
\]
where $\widehat{p}(dx_{n}|y_{0:n})$\ was defined in
(\ref{eq:filteringdistribution}). Finally, we define
\[
P_{k,n}=\frac{D_{k,n}}{D_{k,n}(1)}\quad\text{and}\quad P_{k,n}^{N}%
=\frac{D_{k,n}^{N}}{D_{k,n}^{N}(1)}.
\]
The following estimates are a straightforward consequence of Assumption (A).
For time indices $0\leq k\leq q\leq n$,
\begin{equation}
b_{k,n}=\sup_{x_{k},x_{k}^{\prime}}{\frac{Q_{k,n}(1)(x_{k})}{Q_{k,n}%
(1)(x_{k}^{\prime})}}\leq\rho^{2}\delta^{2},\quad\beta\left(  \frac
{Q_{k,q}(x_{k},dx_{q})Q_{q,n}(1)(x_{q})}{Q_{k,q}(Q_{q,n}(1))}\right)
\leq\left(  1-\rho^{-4}\right)  ^{(q-k)}, \label{eq:contractionEst}%
\end{equation}
and for $0<k\leq q$,
\begin{equation}
M_{k,\eta}(x,dz)\leq\rho^{4}~M_{k,\eta}(x^{\prime},dz)\Longrightarrow
\beta\left(  M_{q,\eta_{q-1}^{N}}\ldots M_{k,\eta_{k-1}^{N}}\right)
\leq\left(  1-\rho^{-4}\right)  ^{q-k+1}. \label{eq:contractionEst2}%
\end{equation}

Several auxiliary results are now presented. For any $\varphi\in
\mathcal{B}(\mathcal{X})$, let%
\begin{equation}
V_{k}^{N}(\varphi)=\eta_{k}^{N}(\varphi)-\Phi_{k}(\eta_{k-1}^{N})(\varphi).
\label{eq:particleV_p}%
\end{equation}
The following is an almost sure Kintchine type inequality\ \cite[Lemma
7.3.3]{delmoral2004}.

\begin{lem}
\label{lem:kintchine} Let $\mathcal{F}_{n}^{N}:=\sigma\left(  \left\{
X_{k}^{(i)};0\leq k\leq n,1\leq i\leq N,\right\}  \right)  $, $n\geq0$, be the
natural filtration associated with the $N$-particle approximation model and
$\mathcal{F}_{-1}^{N}$ be the trivial sigma field. For any $r\geq1$, there
exist a finite (non random) constant $a_{r}$ such that the following
inequality holds for all $k\geq0$ and $\mathcal{F}_{k-1}^{N}\ $measurable
functions $\varphi_{k}^{N}\in\mathcal{B}(\mathcal{X})$\ s.t. $\mbox{\rm
osc}(\varphi_{k}^{N})\leq1$,%
\[
\mathbb{E}\left(  \left\vert \sqrt{N}V_{k}^{N}(\varphi_{k}^{N})\right\vert
^{r}~\left\vert ~\mathcal{F}_{k-1}^{N}\right.  \right)  ^{\frac{1}{r}}\leq
a_{r}.
\]

\end{lem}

This inequality\ may be used to derive the following $\mathbb{L}_{r}$ error
estimate\ \cite[Theorem 7.4.4]{delmoral2004}.

\begin{lem}
\label{lem:lpErrorFilter}For any $r\geq1$, there exists a constant $a_{r}$
such that the following inequality holds for all $k\geq0$ and $\varphi
\in\mathcal{B}(\mathcal{X})$\ s.t. $\mbox{\rm osc}(\varphi)\leq1$,
\begin{equation}
\sqrt{N}\mathbb{E}\left(  \left\vert [\eta_{n}^{N}-\eta_{n}](\varphi
)\right\vert ^{r}\right)  ^{\frac{1}{r}}\leq a_{r}~\sum_{k=0}^{n}%
~b_{k,n}~\beta\left(  \frac{Q_{k,n}}{Q_{k,n}(1)}\right)  .
\label{eq:lpErrorFilter}%
\end{equation}

\end{lem}

A time-uniform bound for (\ref{eq:lpErrorFilter}) may be obtained by using the
estimates in (\ref{eq:contractionEst})-(\ref{eq:contractionEst2}). The final
auxiliary result is the following.

\begin{lem}
\label{lem:Dpn}For time indices $1\leq k\leq n$,
\[
\eta_{k-1}^{N}D_{k-1,n}^{N}(S_{n})=(\eta_{k-1}^{N}Q_{k})(D_{k,n}^{N}(S_{n}))
\]

\end{lem}

\begin{proof}%
\begin{equation*}
\begin{array}{l}
\eta _{k-1}^{N}D_{k-1,n}^{N}(S_{n}) \\
=\int \eta _{k-1}^{N}(dx_{k-1})p^{N}(dx_{0:k-2}\left\vert
x_{k-1},y_{0:k-2}\right. )\left(
\prod\limits_{q=k}^{n}Q_{q}(x_{q-1},dx_{q})\right) S_{n}(x_{0:n}) \\
=\int \eta _{k-1}^{N}(dx_{k-1})Q_{k}(x_{k-1},dx_{k}) \\
\hskip3cm\times p^{N}(dx_{0:k-2}\left\vert x_{k-1},y_{0:k-2}\right. )\left(
\prod\limits_{q=k+1}^{n}Q_{q}(x_{q-1},dx_{q})\right) S_{n}(x_{0:n})%
\end{array}%
\end{equation*}%
The result follows upon noting that
\begin{equation*}
\eta _{k-1}^{N}(dx_{k-1})Q_{k}(x_{k-1},dx_{k})=\eta
_{k-1}^{N}Q_{k}(dx_{k})~M_{k,\eta _{k-1}^{N}}(x_{k},dx_{k-1}).
\end{equation*}
\end{proof}

To prove Theorem \ref{nonasymptheo}, the same semigroup techniques of
\cite[Section 7.4.3]{delmoral2004} are employed.

\begin{proof}  The following decomposition is central
\begin{equation*}
\widehat{\mathcal{S}_{n}}-\mathcal{S}_{n}=\sum_{0\leq k\leq n}\left( \frac{%
\eta _{k}^{N}D_{k,n}^{N}(S_{n})}{\eta _{k}^{N}D_{k,n}^{N}(1)}-\frac{\eta
_{k-1}^{N}D_{k-1,n}^{N}(S_{n})}{\eta _{k-1}^{N}D_{k-1,n}^{N}(1)}\right)
\end{equation*}%
with the convention that $\eta _{-1}^{N}D_{-1,n}^{N}=\eta
_{0}(dx_{0})\prod\limits_{q=1}^{n}Q_{q}(x_{q-1},dx_{q})$, for $k=0$. Lemma %
\ref{lem:Dpn} states that
\begin{equation*}
\eta _{k-1}^{N}D_{k-1,n}^{N}(S_{n})=(\eta
_{k-1}^{N}Q_{k})(D_{k,n}^{N}(S_{n}))
\end{equation*}%
and therefore the decomposition can be also written as
\begin{equation}
\widehat{\mathcal{S}_{n}}-\mathcal{S}_{n}=\sum_{0\leq k\leq n}\left( \frac{%
\eta _{k}^{N}D_{k,n}^{N}(S_{n})}{\eta _{k}^{N}D_{k,n}^{N}(1)}-\frac{\Phi
_{k}(\eta _{k-1}^{N})(D_{k,n}^{N}(S_{n}))}{\Phi _{k}(\eta
_{k-1}^{N})(D_{k,n}^{N}(1))}\right)   \label{eq:decompSn}
\end{equation}%
with the convention $\Phi _{0}(\eta _{-1}^{N})=\eta _{0}$, for $k=0$. Let
\begin{equation*}
\widetilde{S}_{k,n}^{N}=S_{n}-\frac{\Phi _{k}(\eta
_{k-1}^{N})(D_{k,n}^{N}(S_{n}))}{\Phi _{k}(\eta _{k-1}^{N})(D_{k,n}^{N}(1))}.
\end{equation*}%
Then every term in the r.h.s. of (\ref{eq:decompSn}) takes the following
form
\begin{equation}
\frac{\eta _{k}^{N}D_{k,n}^{N}(\widetilde{S}_{k,n}^{N})}{\eta
_{k}^{N}D_{k,n}^{N}(1)}=\frac{\eta _{k}Q_{k,n}(1)}{\eta _{k}^{N}Q_{k,n}(1)}%
\times V_{k}^{N}\left( \overline{D}_{k,n}^{N}(\widetilde{S}%
_{k,n}^{N})\right)   \label{eq:decompSnTerm}
\end{equation}%
where the integral operators $\overline{D}_{k,n}^{N}$ are defined as
follows,
\begin{equation*}
\overline{D}_{k,n}^{N}(S_{n})=\frac{D_{k,n}^{N}(S_{n})}{\eta _{k}Q_{k,n}(1)}.
\end{equation*}%
Finally, using (\ref{eq:decompSn}) and (\ref{eq:decompSnTerm}), $\widehat{%
\mathcal{S}_{n}}-\mathcal{S}_{n}$\ is expressed as
\begin{equation*}
\sqrt{N}\left( \widehat{\mathcal{S}_{n}}-\mathcal{S}_{n}\right)
=I_{n}^{N}(S_{n})+R_{n}^{N}(S_{n})
\end{equation*}%
where the first order term is
\begin{equation*}
I_{n}^{N}(S_{n}):=\sum_{0\leq k\leq n}\sqrt{N}V_{k}^{N}\left( \overline{D}%
_{k,n}^{N}(\widetilde{S}_{k,n}^{N})\right)
\end{equation*}%
and the second order remainder term is
\begin{equation*}
R_{n}^{N}(S_{n}):=\sum_{0\leq k\leq n}\frac{1}{\eta _{k}^{N}\overline{D}%
_{k,n}(1)}\sqrt{N}\left( \eta _{k}-\eta _{k}^{N}\right) \overline{D}%
_{k,n}(1)\times V_{k}^{N}\left( \overline{D}_{k,n}^{N}(\widetilde{S}%
_{k,n}^{N})\right) .
\end{equation*}%
The non-asymptotic variance bound is based on the triangle inequality
\begin{equation}
\mathbb{E}\left\{ N\left( \widehat{\mathcal{S}_{n}}-\mathcal{S}_{n}\right)
^{2}\right\} \leq \left( \mathbb{E}\left\{ I_{n}^{N}(S_{n})^{2}\right\} ^{%
\frac{1}{2}}+\mathbb{E}\left\{ R_{n}^{N}(S_{n})^{2}\right\} ^{\frac{1}{2}%
}\right) ^{2},  \label{eq:nonAsympVarTriaInequality}
\end{equation}%
and bounds are derived below for the individual expressions on the
right-hand side of this equation.\linebreak \qquad Using the fact that $%
\left\{ V_{k}^{N}\left( \overline{D}_{k,n}^{N}(\widetilde{S}%
_{k,n}^{N})\right) \right\} _{0\leq k\leq n}$\ is zero mean and uncorrelated,%
\begin{equation}
\mathbb{E}\left( I_{n}^{N}(S_{n})^{2}\right) =\sum_{0\leq k\leq n}N\mathbb{E}%
\left\{ V_{k}^{N}\left( \overline{D}_{k,n}^{N}(\widetilde{S}%
_{k,n}^{N})\right) ^{2}\right\} .  \label{eq:fisrtOrderError}
\end{equation}%
The following results are needed to bound the right-hand side of (\ref%
{eq:fisrtOrderError}). First, observe that $D_{k,n}^{N}(1)=Q_{k,n}(1)$, and $%
\overline{D}_{k,n}^{N}(1)=\overline{D}_{k,n}(1)$. Now using the
decomposition,
\begin{equation*}
\begin{array}{l}
\overline{D}_{k,n}^{N}(\widetilde{S}_{k,n}^{N})(x_{k}) \\
=\overline{D}_{k,n}(1)(x_{k})\times \int ~\left[
P_{k,n}^{N}(S_{n})(x_{k})-P_{k,n}^{N}(S_{n})(x_{k}^{\prime })\right] ~\Psi
_{Q_{k,n}(1)}(\Phi _{k}(\eta _{k-1}^{N}))(dx_{k}^{\prime }),%
\end{array}%
\end{equation*}%
it follows that
\begin{equation}
\left\Vert \overline{D}_{k,n}^{N}(\widetilde{S}_{k,n}^{N})\right\Vert \leq
b_{k,n}~\mbox{\rm osc}(P_{k,n}^{N}(S_{n}))
\end{equation}%
For linear functionals of the form (\ref{eq:additiveFunctionalSimple}), it
is easily checked that
\begin{equation*}
D_{k,n}^{N}(S_{n})=Q_{k,n}(1)~\sum_{0\leq q\leq k}\left[ M_{k,\eta
_{k-1}^{N}}\ldots M_{q+1,\eta _{q}^{N}}\right] (s_{q})+\sum_{k<q\leq
n}Q_{k,q}(s_{q}~Q_{q,n}(1))
\end{equation*}%
with the convention $M_{k,\eta _{k-1}^{N}}\ldots M_{k+1,\eta _{k}^{N}}=Id$,
the identity operator, for $q=k$. Recalling that $D_{k,n}^{N}(1)=Q_{k,n}(1)$%
, we conclude that
\begin{equation*}
P_{k,n}^{N}(S_{n})=s_{k}+\sum_{0\leq q<k}\left[ M_{k,\eta _{k-1}^{N}}\ldots
M_{q+1,\eta _{q}^{N}}\right] (s_{q})+\sum_{k<q\leq n}\frac{%
Q_{k,q}(Q_{q,n}(1)~s_{q})}{Q_{k,q}(Q_{q,n}(1))}
\end{equation*}%
and therefore
\begin{equation*}
P_{k,n}^{N}(S_{n})=\sum_{0\leq q<k}\left[ M_{k,\eta _{k-1}^{N}}\ldots
M_{q+1,\eta _{q}^{N}}\right] (s_{q})+\sum_{k\leq q\leq n}\frac{%
Q_{k,q}(Q_{q,n}(1)~s_{q})}{Q_{k,q}(Q_{q,n}(1))}
\end{equation*}%
Thus,
\begin{equation}
\begin{array}{l}
\mbox{\rm osc}(P_{k,n}^{N}(S_{n})) \\
\leq \sum_{0\leq q<k}\beta \left( M_{k,\eta _{k-1}^{N}}\ldots M_{q+1,\eta
_{q}^{N}}\right) \mbox{\rm osc}(s_{q})+\sum_{k\leq q\leq n}~\beta \left(
\frac{Q_{k,q}(x_{k},dx_{q})Q_{q,n}(1)(x_{q})}{Q_{k,q}(Q_{q,n}(1))}\right) ~%
\mbox{\rm osc}(s_{q})%
\end{array}
\label{eq:oscBoundParticlePpn}
\end{equation}%
Using the estimates in (\ref{eq:contractionEst}) and (\ref%
{eq:contractionEst2}) for the contraction coefficients, and the estimate in (%
\ref{eq:contractionEst}) for\ $b_{k,n}$, it follows that there exists some
finite (non random) constant $c$ such that the bound
\begin{equation}
\left\Vert \overline{D}_{k,n}^{N}(\widetilde{S}_{k,n}^{N})\right\Vert \leq c
\end{equation}%
holds for any pair of time indexes $k,n$ satisfying $0\leq k\leq n$,
particle number $N$\ and choice of functions $\{s_{k}\}_{0\leq k\leq n}$.
The desired bound for (\ref{eq:oscBoundParticlePpn}) is now obtained by
combining this result with Lemma \ref{lem:kintchine}:%
\begin{align}
\mathbb{E}\left( I_{n}^{N}(S_{n})^{2}\right) & =\sum_{0\leq k\leq n}N\mathbb{%
E}\left( V_{k}^{N}(\overline{D}_{k,n}^{N}(\widetilde{S}_{k,n}^{N}))^{2}%
\right)   \notag \\
& \leq d(n+1)  \label{eq:fisrtOrderError_meanSquare}
\end{align}%
where $d$ is a constant whose value does not depend on $(n,N,S_{n})$%
.\linebreak \qquad Concerning the term $\mathbb{E}\left\{
R_{n}^{N}(S_{n})^{2}\right\} $\ in (\ref{eq:nonAsympVarTriaInequality}).%
\begin{align}
\mathbb{E}\left\{ R_{n}^{N}(S_{n})^{2}\right\} ^{\frac{1}{2}}& \leq
\sum_{0\leq k\leq n}\frac{1}{\sqrt{N}}\mathbb{E}\left\{ \left[ \frac{1}{\eta
_{k}^{N}\overline{D}_{k,n}(1)}\sqrt{N}\left( \eta _{k}-\eta _{k}^{N}\right)
\overline{D}_{k,n}(1)\times \sqrt{N}V_{k}^{N}\left( \overline{D}_{k,n}^{N}(%
\widetilde{S}_{k,n}^{N})\right) \right] ^{2}\right\} ^{\frac{1}{2}}  \notag
\\
& \leq \sum_{0\leq k\leq n}\frac{1}{\sqrt{N}}b_{k,n}\mathbb{E}\left\{ \left[
\sqrt{N}\left( \eta _{k}-\eta _{k}^{N}\right) \overline{D}_{k,n}(1)\times
\sqrt{N}V_{k}^{N}\left( \overline{D}_{k,n}^{N}(\widetilde{S}%
_{k,n}^{N})\right) \right] ^{2}\right\} ^{\frac{1}{2}}  \notag \\
& \leq \sum_{0\leq k\leq n}\frac{1}{\sqrt{N}}b_{k,n}\mathbb{E}\left\{ \left[
\sqrt{N}\left( \eta _{k}-\eta _{k}^{N}\right) \overline{D}_{k,n}(1)\right]
^{4}\right\} ^{\frac{1}{4}}  \notag \\
& \times \mathbb{E}\left\{ \left[ \sqrt{N}V_{k}^{N}\left( \overline{D}%
_{k,n}^{N}(\widetilde{S}_{k,n}^{N})\right) \right] ^{4}\right\} ^{\frac{1}{4}%
}  \notag \\
& \leq \frac{1}{\sqrt{N}}e(n+1)  \label{eq:secondOrderRem_L2}
\end{align}%
where $e$ is a constant whose value does not depend on $(n,N,S_{n})$. The
second line follows from (\ref{eq:contractionEst}) and the third by the
Cauchy-Schwartz inequality. The final line was arrived at by the same
reasoning used to derive bound (\ref{eq:fisrtOrderError_meanSquare}) and
Lemma \ref{lem:lpErrorFilter}. The assertion of the theorem may be verified
by substituting bounds (\ref{eq:fisrtOrderError_meanSquare}) and (\ref%
{eq:secondOrderRem_L2}) into (\ref{eq:nonAsympVarTriaInequality}).
\end{proof}

It is possible to write
\[
\sum\limits_{k=1}^{n}\gamma_{k}^{2}\prod\limits_{i=k+1}^{n}(1-\gamma_{i}%
)^{2}+\sum\limits_{k=2}^{n}\sum\limits_{i=1}^{k-1}\gamma_{i}^{2}%
(1-\gamma_{i+1})^{2}\cdots(1-\gamma_{n})^{2}%
\]
as the sum in (\ref{eq:stepDiscounting}) below.

\begin{lem}
\label{lem:stepDiscounting}Let $\alpha\in(0.5,1]$ and $\gamma_{n}=n^{-\alpha}$
for $n>0$. Then
\begin{equation}
\underset{n\rightarrow\infty}{\liminf}\quad\gamma_{n}^{2}+\sum_{i=1}%
^{n-1}(n+1-i)\gamma_{i}^{2}(1-\gamma_{i+1})^{2}\cdots(1-\gamma_{n})^{2}>0.
\label{eq:stepDiscounting}%
\end{equation}

\end{lem}

\begin{proof}
Let $\left\lfloor a\right\rfloor $ denote the largest integer less than or
equal to $a$. Since the result is obvious for $\alpha=1$, let $\alpha
\in(0.5,1)$.%
\begin{align*}
&  \gamma_{n}^{2}+\sum_{i=\left\lfloor n/2\right\rfloor }^{n-1}(n+1-i)\gamma
_{i}^{2}(1-\gamma_{i+1})^{2}\cdots(1-\gamma_{n})^{2}\\
&  \geq\gamma_{n}^{2}+\gamma_{n}^{2}\sum_{i=\left\lfloor n/2\right\rfloor
}^{n-1}(n+1-i)(1-\gamma_{\left\lfloor n/2\right\rfloor })^{2(n-i)}\\
&  =\gamma_{n}^{2}\left(  \sum_{j=1}^{n+1-\left\lfloor n/2\right\rfloor
}j\lambda_{n}^{j-1}-\frac{1}{(1-\lambda_{n})^{2}}\right)  +\frac{\gamma
_{n}^{2}}{(1-\lambda_{n})^{2}}%
\end{align*}
where $\lambda_{n}=(1-\gamma_{\left\lfloor n/2\right\rfloor })^{2}$ and
\[
\sum_{j>0}j\lambda_{n}^{j-1}=\frac{1}{(1-\lambda_{n})^{2}}.
\]
It may be verified that
\[
\lim_{n\rightarrow\infty}\frac{\gamma_{n}^{2}}{(1-\lambda_{n})^{2}%
}=2^{-2\alpha-2}%
\]
and
\[
\lim_{n\rightarrow\infty}\gamma_{n}^{2}\sum_{j>n+1-\left\lfloor
n/2\right\rfloor }j\lambda_{n}^{j-1}=0.
\]
Hence the result follows.
\end{proof}

\end{document}